\documentclass[11pt]{article}
\usepackage[T1]{fontenc}
\usepackage[utf8]{inputenc}

\usepackage{times}
\usepackage{amsmath}
\usepackage{amsthm}
\usepackage{fullpage}
\usepackage{graphicx} 
\usepackage{makecell}

\usepackage{amssymb}
\usepackage{algorithm}
\usepackage[noend]{algpseudocode}
\usepackage{comment}

\usepackage{tabularx}
\usepackage{multirow}
\usepackage{booktabs}
\usepackage{array}
\usepackage{xurl}
\usepackage{hyperref}
\usepackage{mfirstuc}

\MFUhyphentrue
\MFUnocap{a}
\MFUnocap{an}
\MFUnocap{and}
\MFUnocap{as}
\MFUnocap{at}
\MFUnocap{but}
\MFUnocap{by}
\MFUnocap{for}
\MFUnocap{from}
\MFUnocap{in}
\MFUnocap{into}
\MFUnocap{nor}
\MFUnocap{of}
\MFUnocap{on}
\MFUnocap{onto}
\MFUnocap{or}
\MFUnocap{over}
\MFUnocap{per}
\MFUnocap{the}
\MFUnocap{to}
\MFUnocap{via}
\MFUnocap{versus}
\MFUnocap{vs}
\MFUnocap{with}
\MFUnocap{without}

\usepackage{listings}
\usepackage{caption}

\usepackage[most]{tcolorbox}
\usepackage{xparse}

\usepackage[normalem]{ulem}

\usepackage{xcolor}
\definecolor{keywordcolor}{rgb}{0.7, 0.1, 0.1}   
\definecolor{tacticcolor}{rgb}{0.0, 0.1, 0.6}    
\definecolor{commentcolor}{rgb}{0.4, 0.4, 0.4}   
\definecolor{symbolcolor}{rgb}{0.0, 0.1, 0.6}    
\definecolor{sortcolor}{rgb}{0.1, 0.5, 0.1}      
\definecolor{attributecolor}{rgb}{0.7, 0.1, 0.1} 

\NewTColorBox{PromptBlock}{O{} m}{%
  enhanced, breakable,
  title={#2},               
  fonttitle=\bfseries,      
  coltitle=white,           
  colbacktitle=blue!60!black!80,     
  title filled,             
  colback=gray!4,           
  colframe=blue!60!black,   
  boxrule=0.5pt, arc=2pt,   
  left=2mm, right=2mm, top=1.2mm, bottom=1.2mm,
  #1                        
}

\usepackage[commandnameprefix=ifneeded]{changes}
\definechangesauthor[name={review}, color=blue]{RV}
\definechangesauthor[name={review}, color=red]{LP}

\newtheorem{theorem}{Theorem}

\theoremstyle{definition}
\newtheorem{definition}[theorem]{Definition}

\theoremstyle{plain}

\newcolumntype{Y}{>{\raggedright\arraybackslash}X}

\usepackage{tikz}
\usetikzlibrary{arrows.meta,positioning,shapes.geometric}

\title{{\sc FormaTheoria}: Constructing Large-Scale Lean Theories from Mathematical Literature\\ {\Large \it Toward the Formalization of the Classification of Finite Simple Groups} \\ ~\\ }
\author{
Tianjiao Nie\\Qiuzhen College\\Tsinghua University
\and
Ao Zhang\\Qiuzhen College\\Tsinghua University
\and
Yusen Tang\\Qiuzhen College\\Tsinghua University
\and
Damiano Testa\\Department of Mathematics\\University of Warwick
\and
Shing-Tung Yau\thanks{Correspondence to: Shing-Tung Yau \texttt{<styau@tsinghua.edu.cn>}, Peng Li \texttt{<lipeng@air.tsinghua.edu.cn>} and  Yuan Zhou \texttt{<yuan-zhou@tsinghua.edu.cn>}.}\\Yau Mathematical Sciences Center\\Tsinghua University
\and
Peng Li{$^*$}\\Institute for AI Industry Research (AIR)\\ Tsinghua University
\and
Yuan Zhou{$^*$}\\Yau Mathematical Sciences Center\\Tsinghua University
}
\date{}

\begin{document}

\maketitle

\thispagestyle{empty}

\begin{abstract}
Large-scale formalization of advanced mathematics requires more than translating individual statements: it must reconstruct a coherent theory distributed across heterogeneous sources. This process raises four challenges: discovering implicit dependencies, correcting source defects, preserving semantic fidelity, and reconciling cross-source misalignments. We present {\sc FormaTheoria}, an end-to-end, AI-assisted workflow that coordinates source acquisition, formalization, proof construction, recursive dependency discovery, independent review, and reconciliation, while preserving provenance and protecting approved declarations. A shared agent framework supports long-horizon execution through tool use, context compaction, review-gated termination, section-level source context, and dependency-aware batch parallelization. Applying {\sc FormaTheoria} to major components of the Classification of Finite Simple Groups (CFSG), we construct a machine-checked Lean development extending through the Bender--Suzuki theorem and encompassing the Feit--Thompson Odd Order Theorem, Glauberman's $Z^*$ theorem, and the Brauer--Suzuki theorem. This development verifies an extensive body of deeply interdependent finite-group theory while providing a foundation for continuing the CFSG formalization. An empirical analysis of the code and recorded construction process supports the practical relevance of the identified challenges and illustrates the roles of the corresponding workflow components. Together, these results demonstrate how AI-assisted workflows can reconstruct mathematically significant formal theories from distributed literature by combining language-model agents with formal verification, structured review, and explicit dependency management.
\end{abstract}

\newpage
\tableofcontents

\newpage 
\section{Introduction}
\label{sec:introduction}

Formal mathematics represents definitions, statements, and proofs in languages
whose declarations can be checked by a small trusted kernel. It has enabled
machine-checked proofs of major results and the construction of increasingly
extensive libraries of verified mathematics. Proof
assistants such as Rocq (formerly Coq)~\cite{bertot2013interactive}, Lean~\cite{moura2021lean}, Isabelle/HOL~\cite{nipkow2002isabelle}, and HOL Light~\cite{harrison2009hol} have enabled landmark
human-driven formalization projects, while recent language-model and
search-based systems have made substantial progress in statement formalization
and theorem proving. These developments demonstrate the growing capacity of
formal methods to verify sophisticated mathematics, but most existing projects
and AI systems begin with a designated theorem, a curated proof plan, or an
already constructed formal environment.

The Classification of Finite Simple Groups (CFSG) presents a different challenge.
Its proof is distributed across a large and historically layered literature whose
dependencies, definitions, conventions, and representations must be recovered
and aligned during formalization. We introduce {\sc FormaTheoria},
an AI-assisted workflow for reconstructing the required mathematics from this
literature as a coherent, dependency-connected Lean theory. Starting from target results in the literature, the workflow discovers
and formalizes required dependencies, reconciles incompatible mathematical
interfaces, and integrates approved declarations into an evolving formal
development. The literature provides the mathematical basis for this construction but does
not prescribe a literal reproduction of its organization or presentation.
{\sc FormaTheoria} normally preserves the mathematical content of each source
item, while permitting corrections or reconciliations when the
literature contains defects or cross-source misalignments.

\paragraph{Importance and difficulty of formalizing CFSG.}
The Classification of Finite Simple Groups (CFSG) is one of the defining
achievements of modern mathematics. By the Jordan--H\"older theorem, the simple
composition factors of a finite group are determined, up to isomorphism and
reordering, by the group itself. Finite simple groups therefore serve as the
irreducible building blocks of finite-group theory. CFSG identifies all such
groups: cyclic groups of prime order, alternating groups, groups of Lie type,
and the sporadic groups. The classification consequently provides fundamental
structural infrastructure for finite-group theory and supports applications in
number theory, algebraic and arithmetic geometry, combinatorics, and related
areas \cite{solomon2001brief,smith2018applyingcfsg}. A formal development of
CFSG would make this infrastructure available as a mechanically checked and
reusable foundation for subsequent formal mathematics.

The scale and organization of the proof make CFSG an exceptionally difficult
target for formalization. Its first-generation proof emerged over several
decades through contributions from approximately one hundred mathematicians
and remains dispersed across hundreds of papers and books
\cite{solomon2001brief,gorenstein1994classification}. These sources were
written independently, at different times, and for different purposes; they
frequently leave dependencies implicit and may employ different definitions,
conventions, representations, and assumptions. The history of CFSG also
illustrates the difficulty of verifying such an extensive and distributed
argument. After the classification had been announced, the quasithin case
remained a gap for more than two decades and ultimately required the
two-volume work of Aschbacher and Smith~\cite{aschbacher2004quasithin}. Harada and Solomon later repaired
another omitted case involving a standard component of type
$\widehat{M}_{22}$
\cite{harada2008standard}. The continuing
second-generation program seeks to reorganize the proof into a more coherent
account, but this program itself spans many volumes and remains a major
undertaking. Formalizing CFSG therefore requires the reconstruction of a
coherent, dependency-connected theory from a historically layered body of
literature. A machine-checked development would make dependencies and
hypotheses explicit, document cross-source identifications, expose defects or
incompatibilities encountered during construction, and provide stable formal
interfaces through which later results could invoke the classification. It
would thereby turn one of the most extensive collective achievements in
mathematical history into an auditable and reusable formal theory.

\paragraph{Comparison with major formalization projects.}
Major formalization projects have demonstrated that proof assistants can verify
mathematical arguments whose length, computational content, or structural
complexity extends far beyond the scale of an ordinary paper. Human-led projects
have produced machine-checked proofs of the Four Color Theorem, the
Feit--Thompson Odd Order Theorem, and the Kepler conjecture
\cite{gonthier2008formal,gonthier2013machine,hales2017formal}. More recently,
the Liquid Tensor Experiment showed that Lean and Mathlib can support rapid,
collaborative formalization at the frontier of contemporary mathematical research
\cite{commelin2024abstraction}. In parallel, general-purpose libraries such as
Mizar, the Archive of Formal Proofs, and Mathlib have accumulated extensive bodies
of reusable formal mathematics across a broad range of subjects
\cite{mathlib2020,baanen2025growing}.

These developments exhibit different forms of scale. General-purpose libraries
derive their scale from the breadth of the mathematics accumulated across many
areas and contributors. Theorem-centered projects derive theirs from the depth and
complexity of formalizing a designated result, typically using a curated blueprint
or a small number of systematic expositions. Our development presents a third
setting: constructing a dependency-connected theory whose constituent results and
mathematical interfaces must be recovered from a large and heterogeneous
literature. Table~\ref{tab:intro-scale-comparison} compares representative projects
along three dimensions: approximate code volume, mathematical scope, and the
organization of their source material. Because proof assistants, repository
boundaries, and counting conventions differ, the reported figures provide
order-of-magnitude context rather than normalized measurements.

\begin{table}[t]
\begingroup
\centering
\small
\begin{tabular}{p{0.15\linewidth}p{0.19\linewidth}p{0.26\linewidth}p{0.29\linewidth}}
\toprule
Development & Formal system and scope & Reported scale: code volume and effort & Organization of the mathematical source \\
\midrule
Mathlib \cite{baanen2025growing} & Lean; broad community library & 1.9 million lines (2025); ongoing since 2017 & General-purpose library spanning many mathematical subjects and contributors \\ \addlinespace[5pt]
Odd Order \cite{gonthier2013machine} & Rocq; Feit--Thompson theorem & $\sim$150,000 lines with libraries ($\sim$40,000 for the theorem itself); 6 years & One theorem reconstructed principally from two systematic expositions \\ \addlinespace[5pt]
Flyspeck \cite{hales2017formal,hales2024formal} & HOL Light \& Isabelle; Kepler conjecture & $\sim$500,000 lines; eleven years (2003--2014); $\sim$20 work-years & One proof, including its computational components, organized around a project-specific revised blueprint \\ \addlinespace[5pt]
Liquid Tensor Experiment \cite{commelin2024abstraction} & Lean; one research theorem & $\sim$90,000 lines \cite{buzzard2022beyond}; $\sim$18 months (2020--2022) & One designated theorem developed through a shared blueprint and specification-driven collaboration, by about a dozen mathematicians and several computer scientists \\ \addlinespace[5pt]
{\sc FormaTheoria} (this work) & Lean; several CFSG components forming one dependency-connected theory & More than 994,000 lines in more than 850 files, and 30,298 declarations reachable from Bender--Suzuki\textsuperscript{\ref{fn:snapshot}}; $\sim$7 months & A heterogeneous literature whose sources, dependencies, definitions, and conventions must be discovered and aligned during formalization \\
\bottomrule
\end{tabular}
\caption{A comparison of the scope and scale of representative formal-mathematics developments. Each row reports the code volume and the effort as the project itself reports them; the figures are not normalized across formal systems, and line counts in particular are not comparable between them.}
\label{tab:intro-scale-comparison}
\endgroup
\end{table}

By raw code volume, our development comprises more than 994,000
lines.\protect\footnote{\label{fn:snapshot}Unless otherwise stated, all
quantitative figures reported for our development refer to a snapshot taken when
the proofs were first completed, before subsequent cleanup and additions;
Section~\ref{sec:analysis} analyzes this snapshot.} As a rough indication of
codebase size, this exceeds half of the 1.9 million lines reported for Mathlib in
2025. The more consequential distinction concerns the organization and scope of the
formalized mathematics. Unlike a general-purpose library, our development is
directed toward a connected sequence of structural results in finite-group theory.
Unlike a theorem-centered project, it does not begin with a single curated
formalization blueprint. It instead integrates results recovered from independently
written books and papers into a coherent Lean theory while preserving consistency
across their formal interfaces. The project’s scale therefore reflects both the
volume of its formal development and the breadth and depth of the
dependency-connected theory that it reconstructs. Section~\ref{sec:analysis}
examines this latter dimension more closely through the declaration-dependency
graph and a section-by-section comparison with the Rocq formalization of the Odd
Order Theorem.

\paragraph{Comparison with AI-assisted proof construction.}
Recent AI-assisted theorem-proving systems have substantially expanded the
range of formal proof obligations that can be addressed automatically.
Pretraining on mathematical text, source code, and formal proofs supports the
interpretation of informal mathematics and the generation of proof-assistant
syntax \cite{azerbayev2024llemma}. Reasoning-oriented post-training,
inference-time search, and proof-assistant feedback further support the
decomposition, verification, and iterative repair of multistep arguments
\cite{xin2024deepseekprover,xin2025deepseekprover,
hubert2025alphaproof}. These techniques provide complementary capabilities for
translating mathematical content and constructing mechanically checked proofs.
Benchmarks and theorem-proving systems demonstrate substantial progress on
individual formalization and proof tasks. miniF2F comprises 488 independently
posed competition-level statements, while LeanDojo extracts 98,734
theorem--proof pairs from Mathlib
\cite{zheng2022minif2f,yang2023leandojo}. DeepSeek-Prover, AlphaProof, and
Aristotle employ different combinations of synthetic data, learned reasoning,
verifier feedback, and inference-time proof search
\cite{xin2024deepseekprover,xin2025deepseekprover,
hubert2025alphaproof,achim2025aristotle}. Autoformalization benchmarks extend
the task to translating natural-language statements, but typically treat these
statements as separate inputs \cite{yu2025formalmath}. miniCTX incorporates
repository context spanning tens of thousands of tokens, thereby accounting
for definitions, lemmas, notation, and file-level information beyond the
immediate proof state \cite{hu2025minictx}.

Despite their methodological differences, these settings generally assume that
the target theorem and its surrounding formal environment have already been
specified. In conventional theorem proving, a formal context $\Gamma$ and
target type $\tau$ are given, and the task is to construct a proof term $p$ such
that $\Gamma \vdash p : \tau$. Statement autoformalization additionally
constructs $\tau$ from a supplied natural-language statement, but usually does
so relative to an existing formal context. The principal challenge is therefore
to translate or prove an individual target using definitions, notation, and
supporting results that are already available.

{\sc FormaTheoria} addresses a different task: constructing and extending the
formal context required to state and prove the target results. Starting from a
source item and an existing approved Lean context, the workflow must determine
which additional definitions and theorems are required, locate and interpret
their sources, formalize newly discovered dependencies, reconcile incompatible
mathematical interfaces, and integrate the resulting declarations without
altering approved mathematical content. These activities are interdependent:
proof attempts may expose missing dependencies, source inspection may reveal
defects or misalignments, and reconciliation may create new proof obligations.
Because this process can span many agent interactions and exceed a single model
context, {\sc FormaTheoria} maintains persistent structured project state,
retrieves task-relevant information, and compacts earlier interactions while
preserving essential mathematical decisions.

The distinction therefore lies in the unit of construction and evaluation.
Existing AI-assisted proving systems generally construct an individual statement
or proof relative to a specified target and formal project. {\sc FormaTheoria}
instead produces the elaborated, dependency-resolved, and internally coherent
Lean theory required to state and establish the target results without unsupported
axioms. Semantic review therefore evaluates whether each formalization is mathematically
supported by its sources and whether any departures are justified, rather than
requiring literal fidelity to the source presentation.

\subsection{Challenges of Large-Scale Source-Based Formalization}
\label{sec:challenges}

Large-scale source-based formalization must recover the required mathematical
statements and dependencies as formalization proceeds. Proof construction must
therefore be coordinated with source acquisition, dependency discovery,
cross-source alignment, semantic translation, and source validation. This
setting gives rise to four interdependent challenges, denoted
\textsc{C1}--\textsc{C4}.

\paragraph{\textsc{C1}: Distributed sources and deep dependencies.}
The source corpus comprises numerous books and papers rather than a single
self-contained exposition. Dependencies frequently cross source boundaries, and a
cited result may itself depend on a substantial chain of earlier results. Moreover,
an individual proof in the literature may extend over many pages and invoke
intermediate results whose relevance becomes apparent only during formalization.
Formalizing a target theorem therefore requires the system to discover and
reconstruct a large, evolving dependency context. An omitted or incorrectly
identified dependency may prevent subsequent formalization even when the local
proof argument is otherwise valid.

\paragraph{\textsc{C2}: Misalignments across mathematical sources.}
Because the source corpus is distributed across independently written books and
papers, its constituent sources need not adopt a common system of definitions,
conventions, notation, or representations. Two sources may define the same object
under different conventions, introduce extensionally equivalent formulations
through distinct mathematical constructions, or use the same name for different
objects. A downstream source may also move between such formulations without
explicitly specifying the required identification. Consequently, formalizations
that are individually faithful to their respective sources may not be directly
composable in Lean. Addressing these cross-source misalignments may require bridging equivalences, conversion constructions, or permitted revisions to candidate formalizations.

\paragraph{\textsc{C3}: Semantically incorrect Lean translations.}
Natural-language definitions and theorem statements must be interpreted in their documented mathematical context, since standing assumptions and conventions may be left implicit. A Lean declaration may nevertheless elaborate successfully while misrepresenting the source, for example, by omitting or altering an assumption, mistranslating a quantifier, encoding a different mathematical object, or strengthening or weakening a conclusion. Lean's elaborator establishes well-typedness, not semantic fidelity. Such errors may therefore remain undetected by Lean and propagate to subsequent formalizations. Semantic fidelity must accordingly be assessed independently of elaboration and provability, against both the source item and its surrounding context.

\paragraph{\textsc{C4}: Defects in the mathematical sources.}
The mathematical sources may themselves contain typographical errors, missing
conditions, ambiguous formulations, incomplete arguments, or incorrect statements.
Such defects may initially be difficult to distinguish from errors introduced
during source transcription, dependency reconstruction, translation, or proof
construction. In particular, the failure to prove a faithfully translated statement
may reflect a defect in the source rather than a limitation of the proof procedure.
The formalization workflow must therefore retain sufficient source context to
diagnose the discrepancy and must report cases that cannot be resolved without
mathematical judgment for human investigation.

\medskip

These challenges arise at distinct interfaces but interact closely. A
missing dependency may make a correct translation appear unprovable, a cross-source
misalignment may resemble a translation error, and either a semantic translation
error or a source defect may become apparent only during downstream proof
construction. The workflow introduced in Section~\ref{sec:workflow-overview}
therefore addresses these challenges through distinct but coordinated mechanisms.

\subsection{Our Contributions}
\label{sec:introduction-contributions}
To address these challenges, we introduce {\sc FormaTheoria}. The principal
contributions of this paper are as follows.

\begin{itemize}
\item We develop an end-to-end workflow that coordinates source acquisition, formalization, proof construction, dependency discovery, review, and reconciliation at scale. The workflow preserves provenance, protects approved declarations in a persistent library, and records unresolved issues for human investigation. It is implemented through a shared agent framework designed for long-horizon formalization, combining multi-turn language-model interaction with Lean and other tools, context compaction, review-gated termination, section-level source context, and dependency-aware batch parallelization. Sections~\ref{sec:workflow-overview},~\ref{sec:component-details} and~\ref{sec:implementation-considerations} describe the overall workflow, its constituent components, and the associated implementation considerations, respectively.

\item We construct a large-scale Lean formalization of major components of the CFSG program, extending through the Bender--Suzuki theorem and encompassing the Feit--Thompson Odd Order Theorem, Glauberman's $Z^*$ theorem, and the Brauer--Suzuki theorem. The resulting development verifies an extensive body of deeply interdependent finite-group theory and establishes a foundation for continuing the formalization of the CFSG. The Lean code for this development is available at \texttt{https://github.com/Qiuzhen-CFSG/CFSG}.

\item We conduct an empirical analysis of the resulting Lean code and its recorded construction process. The analysis examines the expansion and structure of source dependencies, the distribution of source defects and cross-source reconciliations, semantic-review outcomes, and the development's scale and construction time in comparison with the human-developed Rocq formalization. These observations support the practical relevance of the challenges identified above and illustrate how the corresponding workflow components address them.
Section~\ref{sec:introduction-observations-implications} summarizes the principal observations and design implications, while Section~\ref{sec:analysis} presents the detailed analysis.

\item We experimentally evaluate three core workflow mechanisms through controlled ablation studies. On the evaluated targets, dependency-aware parallelism reduced proof-construction wall time by 76.1\% ($4.2\times$), whereas section-level context sharing reduced wall time by 26.5\% and total token consumption by 18.0\%. A pilot study of ten source items further found 6/10 correct translations without review, 9/10 with diagnostic review feedback, and 10/10 with a self-improved review rubric. These results quantify the tradeoff between elapsed time and inference cost and provide evidence that context reuse and source-grounded review feedback contribute to construction efficiency and translation fidelity. Section~\ref{sec:ablation} presents these studies.

\end{itemize}

\subsection{Empirical Observations and Design Implications}
\label{sec:introduction-observations-implications}

At the completion-time snapshot analyzed in Section~\ref{sec:analysis}, our workflow had completed Lean formalizations of the Feit--Thompson Odd Order Theorem, Glauberman's $Z^*$ theorem, the Brauer--Suzuki theorem, and the Bender--Suzuki theorem (Appendix~\ref{app:theorems}). These results form a connected theoretical progression rather than four interchangeable targets: Odd Order excludes nonabelian finite simple groups of odd order; Glauberman's $Z^*$ and Brauer--Suzuki convert information about involutions and Sylow $2$-subgroups into global restrictions; and Bender--Suzuki culminates in the classification of finite simple groups with strongly embedded subgroups. The Odd Order development was independently reconstructed principally from Bender--Glauberman and Peterfalvi \cite{bender1994local,peterfalvi2000character}, rather than translated or ported from the existing Rocq formalization \cite{gonthier2013machine}; a shared source decomposition did not determine the same proof route or declaration structure. Extending this chain through Bender--Suzuki required a broader and more heterogeneous literature. The resulting repository contained more than 994,000 lines of Lean in more than 850 files.

Taken together, the resulting formalization and its recorded construction process yield four empirical observations and associated design implications for large-scale formalization. The first three support the practical relevance of Challenges~\textsc{C1}--\textsc{C4} and illustrate the roles of dependency discovery, source correction and reconciliation, and layered semantic review in addressing them. The fourth concerns the quality of the resulting formal artifact.

\paragraph{The scope and difficulty of a large formalization emerge through dependency discovery.}
The results support the practical relevance of Challenge~\textsc{C1}: the required source corpus and mathematical infrastructure cannot necessarily be specified in advance. The three human-supplied sources served only as entry points, while recursive dependency discovery added twelve further sources, accounting for 680 consulted pages, or 65.6\% of the total. Moreover, the resulting dependency graph combines a broad, shallow base with a narrow but exceptionally deep backbone. In the Bender--Suzuki project-only closure, 29.4\% of declarations have depth zero and the median depth is two, whereas the root has depth 458. These findings indicate that source page count alone is a poor measure of formalization difficulty; the discovered dependency structure also determines the required infrastructure, the critical proof paths, and the opportunities for parallel construction. They therefore motivate the workflow's recursive dependency-discovery and dependency-management components (Section~\ref{sec:analysis-c1}).

\paragraph{Source defects and cross-source misalignments require distinct but coordinated responses.}
The development provides evidence for Challenges~\textsc{C2} and~\textsc{C4}. It exposed hidden assumptions, missing hypotheses, an incorrect divisibility condition, and an indexing error, while also encountering incompatible definitions and representations across sources. These two classes of problems require different responses: source defects must be corrected or explicitly justified against the relevant mathematical context, whereas cross-source misalignments require conversion lemmas, compatibility bridges, or reconciled interfaces. The observed modifications were concentrated at representation-heavy interfaces rather than distributed uniformly by source size or declaration count, suggesting that a relatively small number of carefully designed reconciliations can stabilize many downstream declarations. When no source-supported correction or bridge is available, however, the issue must be preserved and escalated for mathematical judgment. These observations motivate the workflow's separation of source correction, cross-source reconciliation, and human intervention (Sections~\ref{sec:issues} and~\ref{sec:analysis-reconciliation})

\paragraph{Semantic fidelity requires review beyond kernel verification.}
The review outcomes illustrate Challenge~\textsc{C3}: Lean elaboration establishes that a term has its stated type, but does not establish that the type represents the source result in its documented context. Eleven of the fourteen PF Part~I sections were returned after their first review, showing that mechanically valid formalizations may still require semantic revision. The workflow therefore combines elaboration and axiom checks with independent review of translated declarations and root theorem statements. This review must also distinguish mistranslation from missing dependencies, source defects, and cross-source misalignments, since these failure modes can produce similar symptoms during proof construction. Automation can perform much of the checking and revision, while human judgment remains necessary when the intended interpretation is ambiguous or a repair lacks clear source support. The evidence thus illustrates the value of layered review that concentrates human attention at unresolved semantic boundaries (Sections~\ref{sec:examples-translation-mistakes} and \ref{sec:analysis-review-reconcile}).

\paragraph{Construction efficiency and artifact quality are distinct objectives.}
In the comparison considered here, the automatically constructed Lean development was completed over a shorter wall-clock period but contains substantially more code than the human-developed Rocq formalization. These figures are not normalized measures of effort, readability, or quality, but they indicate that rapid construction does not by itself produce a concise and reusable library. Large-scale autoformalization should therefore be evaluated not only by the amount of verified mathematics produced, but also by the stability of its abstractions, the transparency of its proofs, and its potential for integration into reusable mathematical libraries (Section~\ref{sec:analysis-rocq}). This observation also motivates the future direction, discussed in Section~\ref{sec:conclusion}, of improving the quality and reusability of AI-generated formalization code.

\section{{\sc FormaTheoria}: Formalization Model and Workflow}
\label{sec:workflow-overview}

This section presents the formalization model and high-level workflow through
which {\sc FormaTheoria} constructs approved Lean formalizations of
natural-language source items.  Figure~\ref{fig:formatheoria-workflow} provides an overview of
{\sc FormaTheoria}, showing how its principal workflow components address the
four challenges and how the underlying agents are executed through a shared
agent framework. The remainder of this section describes the workflow in greater detail. Section~\ref{sec:source-item-formalization} defines source items and their formalizations, together with the approval rules that protect completed mathematical content. Section~\ref{sec:source-defects-and-misalignments} classifies the principal
issues arising from Challenges~\textsc{C2}--\textsc{C4} that we encountered
during autoformalization. Finally, Section~\ref{sec:theory-constructor-workflow-overview} explains how
{\sc FormaTheoria} coordinates source retrieval and transcription, semantic
translation, graph-based proof construction, and on-demand dependency discovery;
how the {\sc Prover} procedure and {\sc ReconcilerAgent} address source defects
and cross-source misalignments, respectively; and how unresolved ambiguities,
errors, and conflicts are escalated for human investigation.

\begin{figure}
    \centering
    \includegraphics[width=\linewidth]{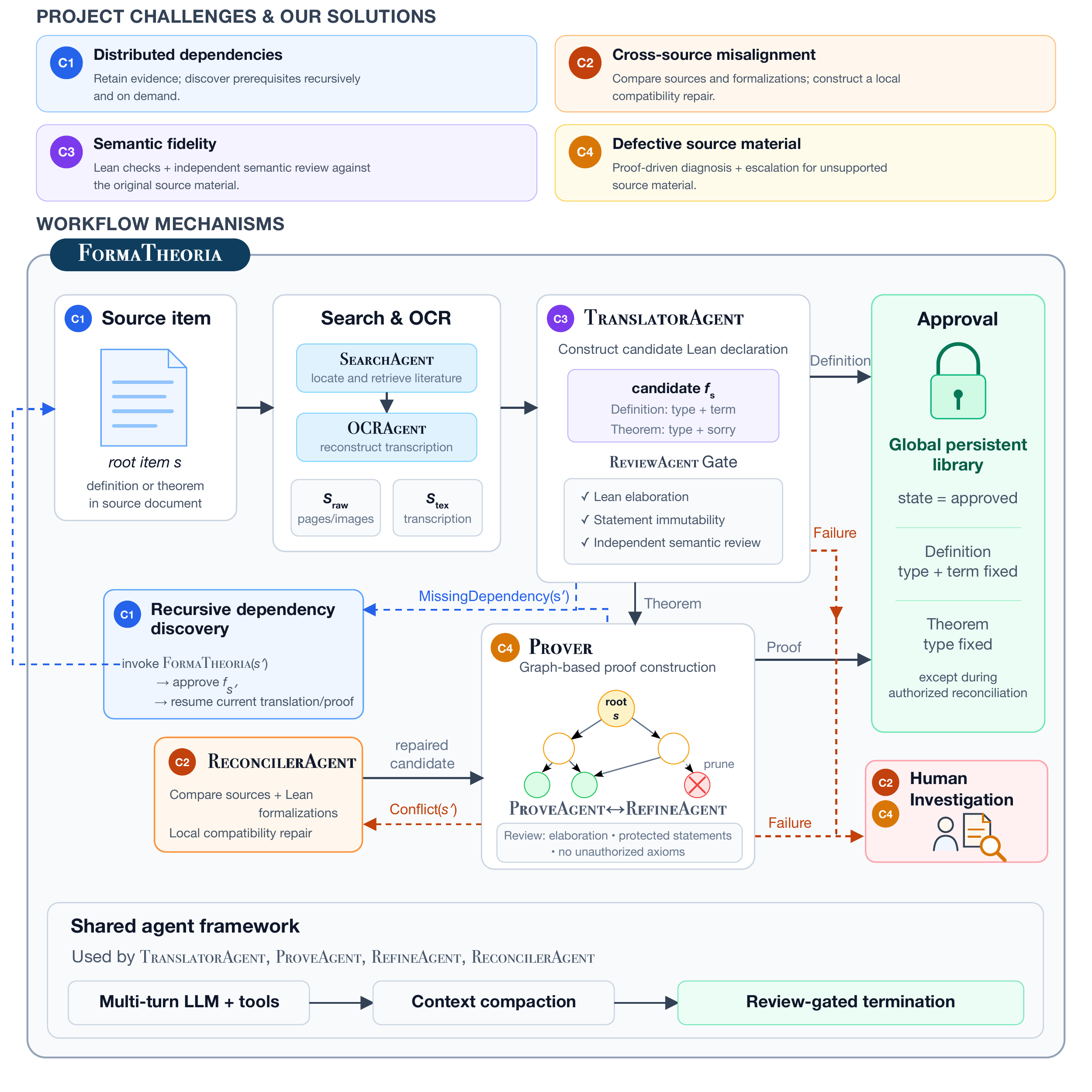}
    \caption{Overview of {\sc FormaTheoria}. Starting from a root source item, the
system retrieves and transcribes the relevant source material, constructs a
candidate Lean declaration, and, for a theorem, performs graph-based proof
construction before independent review and approval. Newly discovered
dependencies are formalized recursively, cross-source conflicts are handled by
{\sc ReconcilerAgent}, and unresolved dependencies, source defects, and
incompatibilities are escalated for human investigation. Approved formalizations are committed to the global persistent
library under the stated immutability rules. The upper panel maps the four principal challenges to their corresponding workflow components, while the lower panel summarizes the shared agent framework through which the agents underlying these components are executed.}
    \label{fig:formatheoria-workflow}
\end{figure}
\subsection{Source Items and Formalizations}
\label{sec:source-item-formalization}

A \emph{source item} is a pointer to a mathematical definition or theorem stated in natural language in mathematical literature supplied by the user or discovered automatically. Here and throughout the paper, the term \emph{theorem} also encompasses results labeled as lemmas, propositions, or corollaries in the source. A source item identifies both the source document and the relevant result, for example, by the document title and theorem number. It thereby specifies what is to be formalized without reproducing the source text itself.

For a source item \(s\), a \emph{formalization} is a tuple
\[
f_s
=
(\mathsf{type}_s,\mathsf{term}_s,\Gamma_s,\mathsf{state}_s),
\]
where \(\Gamma_s\) is a Lean context, \(\mathsf{type}_s\) is a Lean type, and \(\mathsf{term}_s\) is a Lean term satisfying
\[
\Gamma_s\vdash\mathsf{term}_s:\mathsf{type}_s.
\]
If \(s\) is a definition, then \(\mathsf{term}_s\) is the body of the definition, while \(\mathsf{type}_s\) specifies the kind of mathematical object being defined. If \(s\) is a theorem, then \(\mathsf{type}_s\) expresses its assumptions and conclusion, and \(\mathsf{term}_s\) is a proof of that statement under the Curry--Howard correspondence.

The state of a formalization is
\[
\mathsf{state}_s\in
\{\mathsf{candidate},\mathsf{approved}\}.
\]
A candidate formalization is still being constructed or repaired by the
\textsc{TranslatorAgent}, \textsc{ProveAgent}, or \textsc{ReconcilerAgent}.
An approved formalization is regarded as complete. For a definition, approval
means that its Lean body faithfully represents the mathematical definition in
the source. For a theorem, approval means that its Lean type has the same
mathematical meaning as the theorem statement and that its Lean term proves
this type. Note that the mathematical content of a source item is interpreted in its documented
context, including any standing hypotheses, local conventions, or preceding
definitions needed to understand it. Accordingly, assumptions omitted from the
displayed statement may be made explicit in Lean when they are unambiguously
supplied by the surrounding text; doing so does not constitute a departure from
the source.

\paragraph{Global library and the approval rule.} All source items and their formalizations are maintained in a global persistent
library. Once a formalization \(f_s\) is approved, both the type and term of a definition are fixed. For a theorem, the type is fixed, although its proof term may be replaced provided that the resulting code still elaborates. The only exception is that the {\sc ReconcilerAgent}, as introduced later, may perform an authorized revision of an approved formalization when necessary to resolve incompatibilities among source-item formalizations.

\subsection{Classification of the Issues in Autoformalization Process}
\label{sec:source-defects-and-misalignments}

In this section, we classify the principal issues encountered during the
autoformalization process. These issues arise from
Challenges~\textsc{C2}--\textsc{C4}, and {\sc FormaTheoria} incorporates
corresponding mechanisms that attempt to address them automatically. This
classification may not be exhaustive. Issues that the workflow cannot resolve
automatically are escalated for human investigation.

\paragraph{Source defects.}
As noted in Challenge~\textsc{C4}, natural-language mathematical literature may
contain errors that compromise the correctness of its mathematical content. We
distinguish the following two types of source defect:
\begin{enumerate}
  \item[\textsc{[Def1]}] A source item \(s\) contains a typographical or
  mathematical error in a definition or theorem statement, resulting in an
  incorrect definition, hypothesis, or conclusion.

  \item[\textsc{[Def2]}] A source item \(s\) contains a typographical or
  mathematical error in the proof of a theorem, rendering one or more steps of
  the proof invalid.
\end{enumerate}

\paragraph{Source misalignments.}
As noted in Challenge~\textsc{C2}, natural-language mathematical sources may
leave conventions implicit, use the same terminology for different objects, or
present equivalent concepts in formulations that are not directly
interchangeable in Lean. We distinguish the following two types of source
misalignment:
\begin{enumerate}
  \item[\textsc{[Mal1]}] Source items \(s_1\) and \(s_2\) give mathematically
  equivalent definitions of the same object, but the corresponding
  formalizations \(f_1\) and \(f_2\) are not definitionally equal in Lean.

  \item[\textsc{[Mal2]}] Source items \(s_1\) and \(s_2\) use the same name for
  distinct mathematical objects, while a downstream source item \(s_3\)
  implicitly treats these objects as identical.
\end{enumerate}

\paragraph{Incorrect Lean translations.}
This class of issues corresponds to Challenge~\textsc{C3} and concerns the
semantic fidelity of {\sc TranslatorAgent} when translating source definitions
and theorem statements into Lean. An incorrectly translated theorem statement
may render the intended theorem unprovable. An incorrectly translated
definition or theorem statement may also prevent a formalized result from being
applied to a downstream source item. We distinguish the following two types:
\begin{enumerate}
  \item[\textsc{[InT1]}] A theorem source item \(s\) is translated into \(f\), but
  the statement in \(f\) cannot be proved in Lean because the translation is
  incorrect.

  \item[\textsc{[InT2]}] A source item \(s_2\) invokes \(s_1\), but \(f_1\) cannot
  be applied directly in \(f_2\) because \(f_1\) incorrectly translates \(s_1\).
\end{enumerate}

Section~\ref{sec:issues} presents representative source defects and misalignments identified in the literature during our project. Section\ref{sec:examples-translation-mistakes} presents representative examples of incorrect Lean translations produced by {\sc TranslatorAgent}.

\paragraph{Addressing the issues identified above.}
We adopt different strategies for the types of issues identified above. A source
defect of type~\textsc{[Def2]} is typically detected and corrected during the
{\sc Prover} procedure introduced in
Section~\ref{sec:theory-constructor-workflow-overview}. If {\sc Prover} cannot
resolve the defect while preserving the formalized statement, it escalates the
issue for human investigation. Section~\ref{sec:prover} describes this procedure
in detail. Some defects of type~\textsc{[Def1]} can also be detected and
corrected through this procedure. The remaining cases manifest in Lean as issues
of type~\textsc{[InT1]} or~\textsc{[InT2]} and are handled using the
corresponding procedures.

A misalignment of type~\textsc{[Mal1]} can generally be addressed by introducing
a helper lemma relating \(f_1\) and \(f_2\), for example, by establishing the
equivalence of the two formalized definitions. A misalignment of
type~\textsc{[Mal2]} manifests operationally as an issue of type~\textsc{[InT2]}
when a downstream formalization attempts to use \(f_1\) in a context where
\(f_2\) is expected. Such cases are therefore handled in the same manner as
issues of type~\textsc{[InT2]}.
To address issues of types~\textsc{[InT1]} and~\textsc{[InT2]}, we introduce
{\sc ReconcilerAgent}, which diagnoses their underlying causes and constructs
local, source-based repairs. Section~\ref{sec:theory-constructor-workflow-overview}
situates this agent within the overall workflow, while
Section~\ref{sec:reconciler-agent} describes its operation in detail.

\subsection{Workflow Overview}
\label{sec:theory-constructor-workflow-overview}

{\sc FormaTheoria} constructs an approved Lean formalization of a root source
item together with the relevant source items on which it depends. This process
extends beyond the translation of an isolated passage.
As shown in Figure~\ref{fig:formatheoria-workflow}, the system retrieves and transcribes the relevant source material, constructs
the Lean statement or definition, proves the statement when the source item is
a theorem, and checks the resulting artifacts. During translation and proof
construction, the system discovers additional source items that are required
to state or prove the current item. It recursively formalizes these dependencies
and then resumes the interrupted task in the enlarged Lean context.

Dependencies are therefore discovered on demand. The system does not require a
complete dependency structure to be supplied in advance. A dependency may be
identified while translating the mathematical statement or while constructing
its proof. In either case, the \textsc{FormaTheoria} suspends the current
task, invokes itself on the missing source item, and restarts or resumes the
current task after the dependency has been approved.

Previously formalized material may rely on definitions, assumptions, notation,
or representations that are incompatible with the source item currently under
consideration. Such incompatibilities may arise from source defects, source
misalignments, or translation errors and manifest as issues of
type~\textsc{[InT2]}, as defined in
Section~\ref{sec:source-defects-and-misalignments}. When such an incompatibility
is detected, {\sc ReconcilerAgent} compares the relevant sources and Lean
formalizations and attempts to construct a repair.
{\sc ReconcilerAgent} may also address issues of type~\textsc{[InT1]}. In such
cases, {\sc Prover} may determine that the current theorem \(s\) is unprovable
as translated and identify a potential conflict involving a previously
formalized source item \(s'\) that invokes \(s\). {\sc ReconcilerAgent} then
compares the statement of \(s\) with the context in which \(s'\) invokes it and
seeks an appropriate repair. In either case, the repair may revise the current candidate, introduce a justified compatibility construction, replace the proof term of an approved theorem without changing its statement, or revise explicitly authorized theorem statements. The system does not make unauthorized changes to the protected mathematical content of an approved formalization merely to achieve successful Lean elaboration.

The workflow escalates a task when the agents cannot produce a supported
translation, proof, or reconciliation. Human investigation may then examine the
sources and Lean code, edit the formalization, provide additional instructions
to an agent, or invoke the \textsc{FormaTheoria} on another source item.
Thus, human intervention is reserved for cases in which the automated workflow
cannot determine a source-based continuation.

Algorithm~\ref{alg:constructor} summarizes this recursive process. The return
value \(r\) indicates whether an agent has completed its task, found a missing
dependency, detected a conflict, or encountered a failure. The value \(f\)
denotes the current formalization.

\begin{algorithm}[t]
\caption{{\sc FormaTheoria}}
\label{alg:constructor}
\small
\begin{algorithmic}[1]
  \Require source item \(s\)
  \Ensure formalization \(f_s\) of \(s\)

  \State \(S_{\mathrm{raw}}\gets\Call{SearchAgent}{s}\)
  \State \(S_{\mathrm{tex}}\gets\Call{OCRAgent}{S_{\mathrm{raw}}}\)

  \State \((r,f)\gets
  \Call{TranslatorAgent}{s,S_{\mathrm{raw}},S_{\mathrm{tex}}}\)
  \State $f.\mathtt{state} \gets \mathtt{candidate}$

  \If{\(r=\mathsf{MissingDependency}(s')\)}
    \State \(\Call{FormaTheoria}{s'}\) \Comment{Workflow halts and human intervenes if formalization of $s'$ fails}
    \State \Return \(\Call{FormaTheoria}{s}\)
    \Comment{restart translation}
  \ElsIf{\(r=\mathsf{Failure}(e)\)}
    \State \Return \(\Call{ReportToHuman}{e}\)
  \EndIf

  \If{\(s\) is a definition}
    \State $f.\mathtt{state} \gets \mathtt{approved}$
    \State \Return \(f\)
  \EndIf

  \While{\(\mathsf{true}\)} \Comment{additional procedure when $s$ is a theorem}
    \State \((r,f)\gets
    \Call{Prover}{s,S_{\mathrm{raw}},S_{\mathrm{tex}},f}\)

    \If{\(r=\mathsf{MissingDependency}(s')\)}
      \State \(\Call{FormaTheoria}{s'}\) \Comment{Workflow halts and human intervenes if formalization of $s'$ fails}
      \State \textbf{continue}
      \Comment{restart proof after completing the dependency}

    \ElsIf{\(r=\mathsf{Conflict}(s')\)}
      \State Retrieve $S'_{\mathrm{raw}}$, and $S'_{\mathrm{tex}}$ for $s'$ via {\sc SearchAgent} and {\sc OCRAgent}
      \State Let $f_{s'}$ denote the existing formalization for $s'$
      \State \((r_{\textsc{Reconcile}}, f, f_{s'}) \gets \Call{ReconcilerAgent}{s,S_{\mathrm{raw}},S_{\mathrm{tex}},f,\,s',S'_{\mathrm{raw}},S'_{\mathrm{tex}},f_{s'}}\)
      \If{$r_{\textsc{Reconcile}} = \mathsf{Failure}(e)$}
      \State \Return \(\Call{ReportToHuman}{e}\)
      \Else  \Comment{$r_{\textsc{Reconcile}} = \mathsf{Success}$}
      \State \textbf{continue}
      \Comment{continue with the repaired candidate}
      \EndIf

    \ElsIf{\(r=\mathsf{Failure}(e)\)}
      \State \Return \(\Call{ReportToHuman}{e}\)

    \Else  \Comment{$r = \mathsf{Success}$}
        \State $f.\mathtt{state} \gets \mathtt{approved}$
      \State \Return \(f\)
    \EndIf
  \EndWhile
\end{algorithmic}
\end{algorithm}

\medskip

The {\sc FormaTheoria} workflow may invoke the following agents and procedures. Their roles are summarized below, with detailed descriptions provided in Section~\ref{sec:component-details}.

\paragraph{{Search and OCR Agents}.}
The \textsc{SearchAgent} locates and retrieves the literature identified by the
source item. Its output \(S_{\mathrm{raw}}\) contains the relevant source
material in its retrieved form, such as document pages or page images.
The \textsc{OCRAgent} converts the retrieved material into the transcription
\(S_{\mathrm{tex}}\) used by the subsequent agents. The original material
\(S_{\mathrm{raw}}\) is retained alongside the transcription so that the subsequent agents can consult the source evidence when the transcription is ambiguous.

\paragraph{{Translator Agent}.} {\sc TranslatorAgent} converts each source item into a candidate Lean formalization, preserving the complete mathematical meaning. For a definition, it constructs the declaration type and body; for a theorem, it translates the full statement and leaves the proof to the {\sc Prover}. The agent may use approved formalizations and existing library declarations, and it reports any additional source item required for the translation as a missing dependency to {\sc FormaTheoria}. Each candidate must pass both mechanical Lean checks and an independent semantic review that compares the declaration with the original source materials. If the available materials are ambiguous or insufficient to support a faithful translation, the agent reports a failure for human investigation.

\paragraph{{Prover Procedure}.} The {\sc Prover} constructs a proof of a candidate theorem formalization through an iterative, graph-based procedure. It maintains a proof graph whose nodes represent the target theorem and any helper lemmas introduced during proof construction. {\sc ProveAgent} first attempts to discharge the active proof obligation directly. If the obligation is too difficult, {\sc RefineAgent} decomposes it into simpler helper lemmas; if a proposed obligation is determined to be nonviable, the corresponding branch is pruned and the procedure explores another decomposition. Each successful proof must pass mechanical review, including project elaboration, statement immutability, and the absence of unauthorized axioms. The procedure continues until the root theorem is proved or it reports a missing dependency, incompatibility, or unresolved failure to {\sc FormaTheoria}.

\paragraph{{Reconciler Agent}.} {\sc ReconcilerAgent} resolves incompatibilities between a source item under formalization and a previously formalized dependency, corresponding to issues of types~\textsc{[Mal1]} and~\textsc{[Mal2]} introduced in Section~\ref{sec:source-defects-and-misalignments}. It compares the relevant source materials and Lean formalizations to identify the cause of the conflict and construct a local, source-based repair. Depending on the conflict, the agent may revise the current candidate formalization, introduce a justified compatibility definition or lemma, replace the proof term of an approved theorem while preserving its statement, or revise explicitly authorized theorem statements. Each proposed repair must satisfy the applicable approval rules, remain confined to the authorized scope, and preserve the elaboration of all existing downstream formalizations. If no supported repair can satisfy these requirements, the agent reports the unresolved conflict to {\sc FormaTheoria} for human investigation.

\section{Agent Components and Review Procedures}
\label{sec:component-details}

This section details the principal components invoked by the
{\sc FormaTheoria} workflow and the review procedures governing their
execution. Section~\ref{sec:shared-agent-runtime} introduces the shared agent
framework, which supports multi-turn tool use, persistent artifacts, context
compaction, and review-gated termination. Section~\ref{sec:prover} presents the
graph-based {\sc Prover}, in which {\sc ProveAgent} and {\sc RefineAgent}
iteratively construct and decompose proofs while enforcing elaboration and
statement immutability. Section~\ref{sec:translator}
describes how {\sc TranslatorAgent} produces source-faithful Lean declarations
and subjects them to both mechanical checks and independent semantic review.
Finally, Section~\ref{sec:reconciler-agent} explains how
{\sc ReconcilerAgent} constructs local, source-based repairs for
incompatibilities between formalizations while preserving approved mathematical
content and existing downstream developments.

\subsection{Shared Agent Framework}
\label{sec:shared-agent-runtime}

The agents introduced below instantiate a common framework to perform different tasks within the main workflow. The framework supports tool-based interaction with persistent files, allows agents to operate across multiple large language model turns, and permits termination only after an independent review procedure accepts the proposed output. It thereby separates task-specific agent behavior from the underlying mechanisms for tool execution, context management, and output validation.

\paragraph{Input and output.} The agent framework takes as input a task prompt \(p\) and a review agent
\(\mathcal R\). The prompt specifies the task, available artifacts, expected output, and permitted tools, while the review agent defines the conditions under which the task is considered complete.

During execution, the agent may modify designated persistent files through tool calls. At a termination attempt, the framework parses the proposed output $o$ from the files specified by the task contract and returns it only after \(\mathcal R\) accepts it. We therefore model the agent framework as a function with side effects:
\[
    \Call{AgentFramework}{p,\mathcal R}\longrightarrow o.
\]
The files modified during execution constitute its side effects and serve as persistent, structured artifacts for subsequent use by other agents or external tools.

\paragraph{Framework Description.} Algorithm~\ref{alg:agent-loop} gives the conceptual execution procedure. Specifically, the agent first constructs the initial context \(C_0\) from the task prompt. At turn \(k\), it invokes a large language model, which produces a response \(h_k\) and a finite sequence of tool calls \(\mathbf a_k\). If \(\mathbf a_k\neq\varnothing\), the framework executes the permitted calls, obtains their results \(\mathbf r_k\), appends the response, tool calls, and results to the context, and invokes the model again.

A \emph{stop point} is a turn at which \(\mathbf a_k=\varnothing\). At a stop point, the framework parses the proposed output from the files specified by the task contract. Reaching a stop point does not by itself indicate successful completion, because the model may stop prematurely or incorrectly claim that the task is complete. The framework therefore invokes the review agent \(\mathcal R\). If \(\mathcal R\) rejects the output, the framework adds its findings to the context and resumes execution. The framework terminates and returns the output only when \(\mathcal R\) accepts it.

\begin{algorithm}[t]
\caption{Shared Agent Framework}
\label{alg:agent-loop}
\small
\begin{algorithmic}[1]
  \Require task prompt \(p\) and review agent \(\mathcal R\)
  \Ensure reviewed output \(o\) and modified persistent files

  \State \(C_0\gets[p]\)
  \For{\(k=1,2,\ldots\)}
    \State \((h_k,\mathbf a_k)\gets\mathsf{LargeLanguageModel}(C_{k-1})\)

    \If{\(\mathbf a_k=\varnothing\)}
      \State \(o\gets\Call{ParseOutput}{p}\)
      \State \((v,d)\gets\mathcal R(p,o)\)
      \If{\(v=\mathsf{accept}\)}
        \State \Return \(o\)
      \Else
        \State \(C_k\gets C_{k-1}::[h_k,d]\)
        \Comment{continue with review findings}
      \EndIf
    \Else
      \State Execute \(\mathbf a_k\) and obtain results \(\mathbf r_k\)
      \State \(C_k\gets C_{k-1}::[h_k,\mathbf a_k,\mathbf r_k]\)
    \EndIf
  \EndFor
\end{algorithmic}
\end{algorithm}

\paragraph{Review agent.}
The review agent returns either
\(\mathsf{accept}\) or \(\mathsf{reject}\). A rejection is accompanied by
diagnostic feedback \(d\), which explains the failed checks and guides the
next model turn. Each task-specific agent supplies its own review agent,
allowing the same runtime to enforce different completion requirements. The
review procedures used by {\sc ProveAgent} and {\sc RefineAgent} are described
in Section~\ref{sec:prover}.

\paragraph{Context compaction.}
Formalizing long proofs may require an agent to sustain a reasoning trajectory
across many model turns. We denote by \(C_k\) the model context supplied at
iteration \(k\), including the task specification, output contract, and interaction
history accumulated through that iteration, such as relevant model responses and
tool outputs. This model context is distinct from the Lean context \(\Gamma_s\).

To keep the execution within the model's finite context window, the framework
monitors the token length of \(C_k\). When this length exceeds a configured
threshold \(B_{\mathrm{ctx}}\), the framework replaces the earlier portion of the
interaction history with a compact representation, generated either through an
explicit summarization prompt or through the model provider's native compaction
mechanism. The compacted context preserves the task specification, output contract,
established results, unresolved work, and references to persistent artifacts needed
to continue the task. The agent then resumes execution under the same loop using
the compacted context. Algorithm~\ref{alg:agent-loop} omits this compaction
operation for clarity.

\paragraph{Tool calls.}
The common tool allowlist contains the following operations:
\begin{description}
  \item[\texttt{read}]
  Read a designated file or range without modifying it.

  \item[\texttt{edit}]
  Apply a localized modification to a designated existing file and return the
  resulting diff or an error.

  \item[\texttt{write}]
  Create a designated file or replace its contents when the agent is
  authorized to do so.

  \item[\texttt{bash}]
  Execute a shell command and return its standard output, standard error, and
  exit status. This operation supports repository search, Lean compilation,
  audits, and tests.
\end{description}
Each component may impose further restrictions on this common allowlist. In
particular, access to \texttt{bash} does not allow an agent to bypass its file
permissions or acceptance conditions.

\subsection{{\sc Prover}}
\label{sec:prover}

The proof of a substantial mathematical theorem may span many pages in the
source and require tens of thousands of lines of Lean code. Such a task
typically exceeds what a current language-model agent can complete reliably
within a single uninterrupted trajectory. The {\sc Prover} addresses this
long-horizon difficulty by alternating between two agents. {\sc ProveAgent}
attempts to prove one active theorem or helper lemma, while {\sc RefineAgent}
decomposes a proof obligation that is too difficult to solve directly. A
persistent directed acyclic graph records this decomposition and the progress
made across successive agent runs.\footnote{A lemma-graph mechanism implemented in AIM after the release of its initial report provided one point of reference for this design. For the broader natural-language proof workflow and pessimistic verification methodology, see~\cite{liu2025aimathematician,huang2026pessimistic}.}

\paragraph{Input and output.}
The {\sc Prover} has the interface
\[
  (s,S_{\mathrm{raw}},S_{\mathrm{tex}},f)
  \longrightarrow (r,f'),
\]
where \(s\) is the source item, \(S_{\mathrm{raw}}\) is the corresponding raw
literature file, \(S_{\mathrm{tex}}\) is its reconstructed \LaTeX\ source, and
\(f\) is a candidate formalization containing the fixed statement to be
proved. The output \(f'\) is the resulting formalization, and the return status
belongs to
\[
  r\in
  \left\{
    \mathsf{Success},
    \mathsf{Failure}(e),
    \mathsf{MissingDependency}(s'),
    \mathsf{Conflict}(s')
  \right\}.
\]

The status \(\mathsf{Success}\) indicates that the proof has been completed and
has passed the review procedure. A missing dependency or conflict is returned
to {\sc FormaTheoria}, which formalizes the dependency or invokes
{\sc ReconcilerAgent}, respectively. A failure records an error \(e\) that
requires human investigation, as specified in
Algorithm~\ref{alg:constructor}.

\paragraph{Procedure description.} Algorithm~\ref{alg:proof-loop} describes the {\sc Prover}. Its state is
represented by a directed acyclic graph \(G_s\), an active path
\(\mathbf p\), and an active node \(a\), as defined below. Each node $a$ of \(G_s\)
represents either the target formalization \(f\) or a helper lemma introduced
to support its proof, and has a (prover) state
\[
  a.\mathsf{prover\_state} \in \{\mathsf{proposed},\mathsf{proved},\mathsf{pruned}\}.
\]
A proposed node remains to be proved, a proved node has passed the applicable
review procedure, and a pruned node records a proof obligation that
{\sc ProveAgent} has judged not viable.

The root of \(G_s\) is the target formalization \(f\). An edge \(u\to v\)
indicates that \(v\) was introduced as a simpler proof obligation intended to
support the proof of \(u\). The active path \(\mathbf p\) is a selected path
from the root to the node currently being processed, and its final node is the
active node \(a\).

{\sc ProveAgent} first attempts to prove \(a\) directly. If it succeeds, the
algorithm marks \(a\) as proved and returns to its preceding node on the
active path. If the agent returns \(\mathsf{TooHard}(e)\),
{\sc RefineAgent} augments \(G_s\) with one or more helper lemmas and extends
the active path to a newly selected node. If the active node is pruned, the
algorithm also returns to the preceding node and may pursue a different
decomposition. This process continues until the root is proved or the
{\sc Prover} returns a condition to {\sc FormaTheoria}.

The graph serves both as a proof plan and as a record of earlier
attempts. In particular, pruned nodes help prevent {\sc ProveAgent} from
repeatedly following the same unsuccessful decomposition. Whenever the model
context is compacted, the current graph, active path, and active node are
retained in the compacted context. Section~\ref{sec:analysis-long-horizon} quantifies the execution horizons that this mechanism must support.

\begin{algorithm}[t]
\caption{{\sc Prover}}
\label{alg:proof-loop}
\small
\begin{algorithmic}[1]
  \Require source item \(s\), source files \(S_{\mathrm{raw}}\) and
  \(S_{\mathrm{tex}}\), and candidate formalization \(f\)
  \Ensure return status \(r\) and resulting formalization \(f'\)

  \State \(G_s\gets(\{f\},\varnothing)\),  \(f.\mathsf{prover\_state}\gets\mathsf{proposed}\),  \(\mathbf p\gets[f]\)

  \While{\(f.\mathsf{prover\_state}\neq\mathsf{proved}\)}
    \State \(a\gets \text{last element of}~{\mathbf p}\)
    \State \(r\gets
      \Call{ProveAgent}{s,S_{\mathrm{raw}},S_{\mathrm{tex}},f,a}\)

    \If{\(r=\mathsf{Success}\)}
      \State \(a.\mathsf{prover\_state}\gets\mathsf{proved}\)
       \State Remove \(a\) from the end of \(\mathbf p\)
    \ElsIf{\(r=\mathsf{NotProvable}(e)\)}
      \State \(a.\mathsf{prover\_state}\gets\mathsf{pruned}\)      \State Remove \(a\) from the end of \(\mathbf p\)
      \If{\(a=f\)}
        \State \Return \((\mathsf{Failure}(e),f)\)
      \EndIf
    \ElsIf{\(r=\mathsf{MissingDependency}(s')\)}
      \State \Return \((\mathsf{MissingDependency}(s'),f)\)

    \ElsIf{\(r=\mathsf{Conflict}(s')\)}
      \State \Return \((\mathsf{Conflict}(s'),f)\)

    \Else
      \State \(e\gets\Call{FailureMessage}{r}\)
      \Comment{\(r=\mathsf{TooHard}(e)\)}
      \State \((G_s,\mathbf p)\gets
      \Call{RefineAgent}{s,S_{\mathrm{raw}},S_{\mathrm{tex}},f,G_s,\mathbf p,a,e
      }\)
      \If{\(\lvert G_s\rvert>B_{\mathrm{graph}}\)}
        \State \Return
        \((\mathsf{Failure}(\text{graph size limit exceeded}),f)\)
      \EndIf
    \EndIf
  \EndWhile

  \State \(f'\gets\Call{ReadFormalization}{f}\)
  \State \Return \((\mathsf{Success},f')\)
\end{algorithmic}
\end{algorithm}

\paragraph{{\sc ProveAgent}.}
{\sc ProveAgent} runs under the common agent framework described in
Section~\ref{sec:shared-agent-runtime}. It takes
\[
  (s,S_{\mathrm{raw}},S_{\mathrm{tex}},f,a)
\]
as input, where \(a\) is the currently active formalization, and may modify the
authorized Lean source files. It returns one of
\[
  \left\{
    \mathsf{Success},
    \mathsf{TooHard}(e),
    \mathsf{NotProvable}(e),
    \mathsf{MissingDependency}(s'),
    \mathsf{Conflict}(s')
  \right\}.
\]
Its task prompt instructs it to replace the proof placeholder in \(a\) with a
Lean proof supported by the supplied sources, without changing any existing
theorem statement or definition. The review agent used by {\sc ProveAgent}, together with the implementation of its review checks, is specified later in this subsection.

The status \(\mathsf{TooHard}(e)\) indicates that the active obligation should
be decomposed into simpler helper lemmas. By contrast,
\(\mathsf{NotProvable}(e)\) indicates that {\sc ProveAgent} does not regard
the current node as a viable proof obligation, for example, because it appears
mathematically false or because the agent cannot identify a valid proof.

{\sc ProveAgent} returns \(\mathsf{MissingDependency}(s')\) when proving the
active node requires a source item \(s'\) whose formalization is not yet
available. It returns \(\mathsf{Conflict}(s')\) when it detects either of the
following cases:
\begin{enumerate}
    \item The formalization of \(s'\) is available but cannot be used consistently
    in the formalization of \(s\). This case corresponds to an issue of
    type~\textsc{[InT2]} introduced in
    Section~\ref{sec:source-defects-and-misalignments}.
    \item Sufficient evidence suggests that \(s\), as translated, may be
    unprovable, and the agent identifies a downstream source item \(s'\) that
    invokes \(s\). The agent proposes this conflict so that the statement of \(s\)
    can be compared with the context in which \(s'\) invokes it. This case
    corresponds to an issue of type~\textsc{[InT1]}.
\end{enumerate}
When \(\mathsf{Conflict}(s')\) is returned, {\sc FormaTheoria} invokes
{\sc ReconcilerAgent} to investigate the potential incompatibility and construct
an appropriate repair.

\paragraph{{\sc RefineAgent}.}
{\sc RefineAgent} is invoked when {\sc ProveAgent} returns
\(\mathsf{TooHard}(e)\). It receives the source materials, the target
formalization, the current graph \(G_s\), the active path \(\mathbf p\), the
active node \(a\), and the explanation \(e\):
\[
(s,S_{\mathrm{raw}},S_{\mathrm{tex}},f,G_s,\mathbf p,a,e).
\]
Its output is an updated graph and active path:
\[
  (G_s',\mathbf p').
\]

{\sc RefineAgent} analyzes why the current obligation is too difficult,
proposes simpler helper lemmas, and adds the corresponding nodes and edges to
\(G_s\). It then extends the active path and selects its final node as the next
active node. {\sc RefineAgent} also runs under the common agent framework described in
Section~\ref{sec:shared-agent-runtime}, but it uses
the distinct review procedure specified below.

\paragraph{Review procedures.}
At every stop point, the common agent framework  described in
Section~\ref{sec:shared-agent-runtime} invokes the review agent associated
with the current task. The review agent for {\sc ProveAgent} performs the
following mechanical checks:
\begin{enumerate}
  \item the Lean project elaborates successfully;
  \item no existing theorem statement or definition has been modified, as
  detailed under ``Immutable statements'' later in this subsection; and
  \item if the return status is \(\mathsf{Success}\), the resulting
  formalization does not depend on any additional axioms beyond Lean's
  fundamental axioms.
\end{enumerate}
If any check fails, the review agent returns \(\mathsf{reject}\), together with
a description of the failure. The common agent framework adds this feedback to the
context and resumes {\sc ProveAgent}.

The review agent for {\sc RefineAgent} follows the same procedure except that
it omits the first check: the Lean project is not required to elaborate at a
{\sc RefineAgent} stop point. It retains the second and third checks above.
This distinction allows {\sc RefineAgent} to construct an intermediate proof
decomposition before all newly proposed obligations have been completed.

\paragraph{Immutable statements.}
The second review condition requires a stronger safeguard than an instruction
in the agent prompt. A language-model agent may inadvertently modify an
existing statement, and Lean's metaprogramming facilities make source-level
audits alone insufficient for reliably detecting such changes. The review procedure
therefore uses kernel-level declaration information to enforce statement
immutability.

Before {\sc ProveAgent} or {\sc RefineAgent} begins, the agent snapshots the
current Lean repository and records a statement hash for every existing
top-level Lean declaration. At each stop point, it computes the hashes again
and compares them with the snapshot. If any hash differs, the review agent
rejects the output and reports the modified declarations. The agent loop then
resumes with an explicit instruction to restore the affected statements and
definitions. This check restricts the agents to constructing proofs and
permitted helper declarations without altering existing mathematical
statements.

\subsection{{\sc TranslatorAgent}}
\label{sec:translator}

Translating a mathematical definition or theorem statement into Lean requires
preserving its mathematical meaning while expressing it through Lean's type
system and the APIs available in the project. An incorrect translation may make
the intended theorem unprovable or cause downstream formalizations to rely on a
statement that differs from the source. The principal challenge is therefore
semantic fidelity rather than syntactic correctness alone. The
{\sc TranslatorAgent} addresses this challenge by combining Lean-based
mechanical checks with an independent, LLM-based semantic review.

\paragraph{Input and output.}
The {\sc TranslatorAgent} has the interface
\[
  (s,S_{\mathrm{raw}},S_{\mathrm{tex}})
  \longrightarrow (r,f),
\]
where \(s\) is the source item, \(S_{\mathrm{raw}}\) is the corresponding raw
literature file, and \(S_{\mathrm{tex}}\) is its reconstructed \LaTeX\ source.
The output \(f\) is the candidate formalization of \(s\), and the return status
belongs to
\[
  r\in
  \left\{
    \mathsf{Success},
    \mathsf{Failure}(e),
    \mathsf{MissingDependency}(s')
  \right\}.
\]
The status \(\mathsf{Success}\) indicates that \(f\) has passed the applicable
Lean checks and has been accepted by the semantic review as representing the
same mathematical content as \(s\). The status
\(\mathsf{MissingDependency}(s')\) indicates that the formalization of another
source item \(s'\) is required to express \(s\). The
{\sc FormaTheoria} first formalizes \(s'\) and then restarts the
translation of \(s\). The status \(\mathsf{Failure}(e)\) records an unresolved
error \(e\) and causes the task to be reported for human investigation, as in
Section~\ref{sec:prover}.

\paragraph{Procedure description.}
The {\sc TranslatorAgent} instantiates the common agent framework described in
Section~\ref{sec:shared-agent-runtime}. Its task prompt instructs it to inspect
both \(S_{\mathrm{raw}}\) and \(S_{\mathrm{tex}}\), identify the mathematical
content associated with \(s\), and construct the corresponding Lean
declaration.

If \(s\) is a definition, the type and body of \(f\) must faithfully represent
the definition in the source. If \(s\) is a theorem, the type of \(f\) must
faithfully represent the complete theorem statement, including its parameters,
assumptions, quantifiers, and conclusion, while its proof term is left as
\verb|sorry| for subsequent construction by the {\sc Prover}. During
translation, the agent may invoke Lean to diagnose and repair elaboration
errors. It may use existing library declarations and approved formalizations,
but it may not modify them. If an additional source item is required to state
the current item, the agent returns
\(\mathsf{MissingDependency}(s')\). If the source material is insufficient or
ambiguous and the agent cannot establish a supported translation, it returns
\(\mathsf{Failure}(e)\) rather than selecting an unsupported interpretation.

\paragraph{Review procedure.}
The review procedure for {\sc TranslatorAgent} is particularly important because
an incorrect translation can propagate errors to downstream formalizations.

At each stop point, the common agent framework invokes the review procedure
associated with {\sc TranslatorAgent}. The procedure first performs the following
mechanical checks:
\begin{enumerate}
  \item the Lean project elaborates successfully; and
  \item no existing theorem statement or definition has been modified.
\end{enumerate}
The second condition is enforced by the statement-immutability mechanism
described in Section~\ref{sec:prover}.

These checks establish that Lean accepts the proposed declaration and that
previously formalized mathematical content remains unchanged. However, they do
not establish that \(f\) has the same mathematical meaning as \(s\). The review
procedure therefore also invokes an independent {\sc ReviewAgent} to assess the
semantic fidelity of the translation.

The {\sc ReviewAgent} receives \(s\), \(S_{\mathrm{raw}}\),
\(S_{\mathrm{tex}}\), and the proposed formalization \(f\). Its prompt instructs
it to compare the source with the Lean declaration and identify any semantic
discrepancies. The prompt includes a checklist of recurrent translation errors,
including:
\begin{enumerate}
  \item moving part of a theorem's conclusion into its assumptions;
  \item omitting conditions stated or implicitly required by the source;
  \item confusing existential and universal quantification;
  \item using a Mathlib declaration whose semantics differ subtly from the
  intended mathematical concept; and
  \item replacing a specific typeclass assumption required by the source with a
  different or unjustifiably general assumption.
\end{enumerate}

The checklist is initially authored by humans based in part on recurrent
semantic errors identified during the early stages of the project. After a batch
of {\sc TranslatorAgent} trajectories and the corresponding human interventions
has been collected, a {\sc SelfImproveAgent} updates the checklist to incorporate
additional recurrent failure patterns. The prompt templates used by these agents
are provided in the appendix.

At each stop point, the {\sc ReviewAgent} is invoked with a fresh context. It
receives neither findings from previous reviews nor the internal reasoning
trajectory of {\sc TranslatorAgent}. This separation helps ensure that the
semantic review remains independent of the process that produced the candidate
formalization. If either a mechanical check or the semantic review fails, the
review procedure returns \(\mathsf{reject}\) with concrete findings. The common
agent framework adds these findings to the context and resumes
{\sc TranslatorAgent}. The translation terminates successfully only after all
applicable checks accept the proposed formalization.

\subsection{{\sc ReconcilerAgent}}
\label{sec:reconciler-agent}

{\sc ReconcilerAgent} addresses incompatibilities between a source item under
formalization and previously formalized material, corresponding to issues of types~\textsc{[Mal1]} and~\textsc{[Mal2]} in
Section~\ref{sec:source-defects-and-misalignments}. When such an incompatibility is
detected, the agent compares the relevant source materials and Lean
formalizations and attempts to construct a source-based repair.

\paragraph{Input and output.}
{\sc ReconcilerAgent} has the interface
\[
  (s,S_{\mathrm{raw}},S_{\mathrm{tex}},f,\,
   s',S'_{\mathrm{raw}},S'_{\mathrm{tex}},f')
  \longrightarrow
  (r,\widetilde f,\widetilde f'),
\]
where \(s\) is the source item currently being formalized and \(s'\) is a source
item on which \(s\) depends. The objects \(S_{\mathrm{raw}}\) and
\(S_{\mathrm{tex}}\) are, respectively, the raw literature file containing
\(s\) and its reconstructed \LaTeX\ source; the primed objects provide the
corresponding materials for \(s'\). Finally, \(f\) and \(f'\) are the current
formalizations of \(s\) and \(s'\).

The agent is invoked under the following two cases, corresponding to $\mathtt{Conflict}$ cases returned by {\sc ProveAgent}, as described in Section~\ref{sec:prover}:
\begin{enumerate}
    \item \(f\) cannot use \(f'\) as required by the dependency
of \(s\) on \(s'\); for example, some arguments required by \(f'\) may not be
constructible in the context of \(f\), or the type or term obtained from \(f'\)
may differ from what \(f\) expects.
    \item {\sc Prover} determines that theorem \(s\) is unprovable as translated and identifies a potential conflict involving a previously
formalized source item \(s'\) that invokes \(s\).
\end{enumerate}

The outputs \(\widetilde f\) and \(\widetilde f'\) are the formalizations
resulting from the attempted reconciliation, and
\[
  r\in\left\{
    \mathsf{Success},
    \mathsf{Failure}(e)
  \right\}.
\]
The status \(\mathsf{Success}\) indicates that the incompatibility has been
resolved and that the resulting Lean project has passed the applicable review
procedure. The status \(\mathsf{Failure}(e)\) records an unresolved conflict
together with an explanation \(e\). {\sc FormaTheoria} reports such a
failure for human investigation.

\paragraph{Reconciliation procedure and authorized revisions.}
{\sc ReconcilerAgent} instantiates the common agent framework described in
Section~\ref{sec:shared-agent-runtime}. Its task prompt instructs it to inspect
both source items, their surrounding source materials, and their Lean
formalizations before determining the cause of the incompatibility. Depending
on the conflict, the agent may revise the current candidate formalization,
introduce a justified compatibility definition or lemma, or replace the proof
term of an approved theorem while preserving its statement. All modifications must comply with the approval rules, under which the {\sc ReconcilerAgent} is additionally \emph{authorized} to revise the statements (or definition bodies) of both \(f\) and \(f'\).

No uniform repair rule suffices for all conflicts. The agent is therefore instructed to seek
a best-effort, source-based resolution subject to two principles:
\begin{enumerate}
  \item the repair should be local and minimal, avoiding unnecessary changes to
  the surrounding formalization; and
  \item all existing downstream uses of the affected formalizations must
  continue to elaborate.
\end{enumerate}
If the available sources do not support a repair satisfying these requirements,
the agent returns \(\mathsf{Failure}(e)\) rather than altering protected
mathematical content merely to make the project elaborate.

\paragraph{Review procedure.}
At each stop point, the common agent framework invokes the review procedure
associated with {\sc ReconcilerAgent}. The procedure performs the following
mechanical checks:
\begin{enumerate}
  \item the Lean project elaborates successfully, and all existing downstream
  formalizations continue to elaborate without modification to definitions or theorem statements; and
  \item all changes are confined to the declarations and auxiliary constructions
  authorized by the reconciliation task, and every approved formalization
  continues to satisfy the applicable immutability rules.
\end{enumerate}

\section{Implementation Considerations}
\label{sec:implementation-considerations}

This section describes implementation choices that improve the efficiency of the
{\sc FormaTheoria} workflow without changing its formalization model or approval
rules.

\subsection{Section-Level Context Sharing}
\label{sec:context-sharing}

Context sharing allows information accumulated during one agent execution to be
carried forward to subsequent executions on related tasks. This mechanism is
particularly useful in source-based formalization because related source items
often rely on the same mathematical background, literature dependencies, and
Mathlib declarations. Reusing this information reduces the token and time cost of
reconstructing the relevant context for each source item independently.

To avoid confusion with the Lean context \(\Gamma_s\), we use \emph{model context}
to refer to the information supplied to an agent during its execution. Instances
of {\sc TranslatorAgent} and {\sc ProverAgent} that process related source items
reuse a common, evolving model context. We define relatedness at the section level:
source items are treated as related when they occur in the same section of a book
or paper. This boundary reflects the organization of mathematical literature,
where items within a section typically share notation, standing assumptions,
definitions, intermediate results, and relevant external dependencies. At the
same time, it limits the inclusion of information from unrelated parts of the
source corpus.

Accordingly, when Algorithm~\ref{alg:agent-loop} is invoked for a source item, its
initial context \(C_0\) is initialized from the most recent model context available
for that section at the time of invocation. The resulting section-level model
context retains information accumulated while processing earlier source items,
including relevant source passages, previously identified literature dependencies,
and useful Mathlib declarations. Discoveries made during one agent execution can
therefore be reused by subsequent {\sc TranslatorAgent} and {\sc ProveAgent}
instances working within the same section. Section-level context sharing thus
amortizes the token cost of reconstructing mathematical background and discovering
literature and Mathlib dependencies across multiple translation and proof-construction tasks.

\subsection{Dependency-Aware Batch Parallelization}
\label{sec:batch-parallelization}

Independent source items can be formalized concurrently, thereby reducing the
wall-clock time required to process a large source corpus. Parallel execution must,
however, respect dependencies discovered during formalization and coordinate
modifications to the shared persistent library. We therefore combine batch-level
parallelism with dependency-aware synchronization and globally serialized
reconciliation.

In our implementation, source items are organized into batches, each typically
corresponding to a section of a book or paper. Subject to the dependencies known at
the time of scheduling, {\sc FormaTheoria} is invoked in parallel for all source
items in a batch. Because the complete dependency structure is generally unavailable
in advance, an invocation may discover during statement construction or proof
construction that its source item depends on another item currently being processed
by a concurrent invocation.

To coordinate such cases, the system records the active {\sc FormaTheoria}
invocation, if any, associated with each source item. When an invocation discovers
a dependency, it first checks this record. If the dependency is already being
processed, the invocation does not launch a duplicate task. Instead, it waits for
the existing invocation to complete and then resumes using the resulting
formalization. If the dependency is not being processed and has not yet been
formalized, the system invokes {\sc FormaTheoria} for that dependency through the
ordinary recursive workflow. This mechanism allows independent tasks to proceed
concurrently while ensuring that each newly discovered dependency is formalized
only once.

Parallel execution also creates potential write conflicts in the shared persistent
library. In particular, {\sc ReconcilerAgent} may modify existing Lean declarations,
including definitions, theorem statements, and proofs, subject to the approval rules
described in Section~\ref{sec:source-item-formalization}. We therefore treat
reconciliation as a globally exclusive operation. At most one
{\sc ReconcilerAgent} may be active at any time. Before it begins, all other
formalization processes are paused and remain suspended until the reconciliation
has completed and its changes have been committed to the persistent library. The
suspended processes then resume against the updated library state. This
serialization prevents other processes from reading or modifying declarations
while they are being reconciled.

\section{Empirical Findings from Large-Scale Formalization}
\label{sec:analysis}

This section analyzes a completion-time snapshot of the formalization.  The snapshot contains more than 994,000 lines of Lean in more than 850 files.  More than 10\% of these lines lie outside the combined dependency closures of the Odd Order theorem and the Bender--Suzuki theorem.  Repository size should therefore not be identified with the size of either proof: the repository also contains declarations supporting other results and intermediate developments that are not used by either root theorem.  Unless stated otherwise, all counts below refer to this snapshot and do not include later edits.  We use these measurements to identify structural phenomena, but avoid treating lines of code, tokens, or elapsed time as direct measures of mathematical difficulty.

\subsection{Distributed Sources and Deep Dependencies}
\label{sec:analysis-c1}

This subsection provides empirical evidence for Challenge~\textsc{C1} (Section~\ref{sec:challenges}).  We examine the challenge at three levels: the expansion of the external source corpus, the weak relationship between source length and formalization burden, and the depth and connectivity of the resulting formal dependency graph.

\subsubsection{Distributed Source Discovery}

The formalization draws on 15 books and papers, covering 1,037 consulted pages (Table~\ref{tab:source-literature}).  The three sources supplied by a human collaborator---the books of Bender--Glauberman and Peterfalvi \cite{bender1994local,peterfalvi2000character}, together with the fourth volume of Gorenstein--Lyons--Solomon \cite{gorenstein1999classification4}---account for 357 pages (34.4\%).  Following dependencies led the workflow to twelve additional sources covering 680 pages (65.6\%).  Thus, dependency discovery expanded the corpus fivefold, from three entry points to fifteen sources, and added about 1.9 times as much consulted material by page count as was initially supplied.\footnote{Page density varies across books and papers, so these totals are a coarse measure of source coverage rather than a normalized measure of mathematical content.} The significance for Challenge~\textsc{C1} is not merely that the corpus became larger.  Almost two-thirds of the consulted pages lay outside the initially supplied sources, so the boundary of the formalization task was not known in advance.  The workflow had to discover the source corpus while simultaneously constructing the formal dependency graph: when a proof reached a dependency across a source boundary, progress could require locating the relevant passage, identifying and reconstructing its own prerequisites, and inserting the resulting declarations before the downstream item could resume.  A single omitted or misidentified dependency could therefore block a long chain even when the local argument at its endpoint was otherwise valid.  The human-supplied books served as entry points rather than a closed specification; empirically, the challenge was the open-ended coupling of source discovery and proof construction.

\begin{table}[t]
  \centering
  \footnotesize
  \setlength{\tabcolsep}{3.5pt}
  \renewcommand{\arraystretch}{1.32}
  \begin{tabularx}{\linewidth}{
      >{\centering\arraybackslash}p{0.035\linewidth}
      >{\raggedright\arraybackslash}p{0.255\linewidth}
      >{\raggedright\arraybackslash}X
      >{\raggedright\arraybackslash}p{0.205\linewidth}
      }
    \toprule
    No. & Authors                                                                                                                   & Source & Pages Consulted \\
    \midrule
    1   & Helmut Bender and George Glauberman
        & Local Analysis for the Odd Order Theorem \cite{bender1994local}
        & 1--152                                                                                                                                               \\
    2   & Thomas Peterfalvi
        & Character Theory for the Odd Order Theorem \cite{peterfalvi2000character}
        & 1--150                                                                                                                                               \\
    3   & Daniel Gorenstein, Richard Lyons, and Ronald Solomon
        & The Classification of the Finite Simple Groups, Number 4 \cite{gorenstein1999classification4}
        & 19--64, 245--253                                                                                                                                     \\
    4   & Marshall Hall, Jr.
        & The Theory of Groups \cite{hall1959theory}
        & 200--217                                                                                                                                             \\
    5   & I. Martin Isaacs
        & Character Theory of Finite Groups \cite{isaacs1976character}
        & 1--11, 78--121, 262--285                                                                                                                             \\
    6   & Daniel Gorenstein
        & Finite Groups \cite{gorenstein1968finite}
        & 102--112, 217--238, 270--280                                                                                                                         \\
    7   & Graham Higman
        & Suzuki 2-Groups \cite{higman1963suzuki}
        & 79--96                                                                                                                                               \\
    8   & Bertram Huppert
        & Endliche Gruppen I \cite{huppert1967endliche}
        & 135--142, 153--273, 448--465, 526--539                                                                                                               \\
    9   & Bertram Huppert and Norman Blackburn
        & Finite Groups III \cite{huppert1982finite3}
        & 28--35, 172--195, 219--227, 246--291                                                                                                                 \\
    10  & Michio Suzuki
        & Group Theory II \cite{suzuki1986group2}
        & 120--166, 266--286                                                                                                                                   \\
    11  & Daniel Gorenstein, Richard Lyons, and Ronald Solomon
        & The Classification of the Finite Simple Groups, Number 2 \cite{gorenstein1996classification2}
        & 53--61, 72--78, 96--106, 140--147                                                                                                                    \\
    12  & Christoph Hering
        & On Finite Groups Operating Doubly Transitively on Their Involutions \cite{hering1971finite}
        & 456--458                                                                                                                                             \\
    13  & Walter Feit
        & The Representation Theory of Finite Groups \cite{feit1982representation}
        & 140--247, 450--475                                                                                                                                   \\
    14  & Thomas Peterfalvi
        & Le th\'eor\`eme de Bender--Suzuki I \cite{peterfalvi1986bendersuzuki}
        & 193--197                                                                                                                                             \\
    15  & Helmut Bender
        & Transitive Gruppen gerader Ordnung, in denen jede Involution genau einen Punkt festl\"a{\ss}t \cite{bender1971transitive}
        & 527--554                                                                                                                                             \\
    \bottomrule
  \end{tabularx}
  \caption{Source literature consulted in the formalization.}
  \label{tab:source-literature}
\end{table}

These corpus-level counts capture the breadth of Challenge~\textsc{C1} across sources, but source discovery also exposes a second layer of dependency reconstruction within individual texts. Mathematical sources differ substantially in their expository organization. Some (e.g.~\cite{bender1994local}) adopt a declaration-oriented style, in which definitions, lemmas, and theorems are presented as locally self-contained units with explicit hypotheses and conclusions. Others (e.g.~\cite{peterfalvi2000character}) use a running-context style, in which a section develops a mathematical configuration through consecutively numbered items whose roles may vary: some state results, while others introduce notation, record intermediate observations, or continue the surrounding discussion. Later results can consequently inherit objects, assumptions, and intermediate facts from an extended preceding context. Locating the relevant source passage is then only the first step: before that passage can be formalized, the workflow must recover the evolving local context that gives it complete logical content and determine which parts of that context belong in the formal dependency graph. Thus, a newly discovered source may enlarge not only the document corpus but also a dependency context whose boundary becomes visible only as the downstream formalization proceeds.

\subsubsection{Source Length and Formalization Burden}

\begin{table}[t!]
  \centering
  \scriptsize
  \setlength{\tabcolsep}{2.6pt}
  \renewcommand{\arraystretch}{1.03}
  \resizebox{0.9\textwidth}{!}{%
  \begin{tabular}{lrrr@{\hspace{7pt}}|@{\hspace{7pt}}lrrr@{\hspace{7pt}}|@{\hspace{7pt}}lrrr}
    \toprule
    \multicolumn{4}{c@{\hspace{7pt}}|@{\hspace{7pt}}}{Bender--Glauberman (BG)} &
    \multicolumn{4}{c@{\hspace{7pt}}|@{\hspace{7pt}}}{Peterfalvi (PF) Part~I}  &
    \multicolumn{4}{c}{Peterfalvi (PF) Part~II and appendices}                                                                                            \\
    Section                                & Local  & Pages & Library & Section & Local  & Pages & Library & Section      & Local  & Pages & Library \\
    \midrule
    Sec.~1                                 & 4,009  & 8     & 19,528  & Sec.~1  & 12,236 & 2     & 58,913  & II.1.1       & 6,510  & 3     & 15,362  \\
    Sec.~2                                 & 174    & 8     & 19,531  & Sec.~2  & 7,639  & 5     & 28,363  & II.1.2       & 5,654  & 1     & 97,918  \\
    Sec.~3                                 & 24,842 & 16    & 48,636  & Sec.~3  & 19,287 & 5     & 56,918  & II.1.3       & 13,885 & 4     & 66,484  \\
    Sec.~4                                 & 14,140 & 11    & 38,973  & Sec.~4  & 11,882 & 6     & 56,833  & II.2         & 29,261 & 7     & 106,715 \\
    Sec.~5                                 & 6,838  & 5     & 38,473  & Sec.~5  & 19,119 & 5     & 48,272  & II.3.1       & 11,115 & 3     & 108,413 \\
    Sec.~6                                 & 6,605  & 6     & 39,408  & Sec.~6  & 25,513 & 8     & 59,113  & II.3.2       & 2,573  & 1     & 54,696  \\
    Sec.~7                                 & 5,287  & 6     & 42,270  & Sec.~7  & 6,648  & 6     & 29,574  & II.3.3       & 4,699  & 3     & 67,181  \\
    Sec.~8                                 & 7,176  & 3     & 42,432  & Sec.~8  & 19,291 & 6     & 75,104  & II.4.1       & 2,743  & 1     & 2,291   \\
    Sec.~9                                 & 6,425  & 5     & 42,657  & Sec.~9  & 73,314 & 8     & 69,647  & II.4.2       & 8,397  & 6     & 39,994  \\
    Sec.~10                                & 16,858 & 11    & 51,573  & Sec.~10 & 42,717 & 6     & 78,319  & II.4.3       & 5,507  & 3     & 53,663  \\
    Sec.~11                                & 5,005  & 3     & 51,547  & Sec.~11 & 24,995 & 5     & 78,431  & II.4.4       & 3,122  & 3     & 58,562  \\
    Sec.~12                                & 30,374 & 14    & 52,153  & Sec.~12 & 19,884 & 6     & 78,663  & Appendix~I   & 2,572  & 2     & 97,270  \\
    Sec.~13                                & 11,952 & 8     & 51,954  & Sec.~13 & 54,809 & 12    & 80,029  & Appendix~II  & 3,230  & 2     & 98,308  \\
    Sec.~14                                & 20,985 & 12    & 52,595  & Sec.~14 & 32,388 & 6     & 97,151  & Appendix~III & 2,953  & 5     & 48,436  \\
    Sec.~15                                & 23,516 & 6     & 52,850  &         &        &       &         & Appendix~IV  & 11,862 & 6     & 58,787  \\
    Sec.~16                                & 11,074 & 12    & 52,931  &         &        &       &         &              &        &       &         \\
    Appendix~C                             & 8,068  & 6     & 50,178  &         &        &       &         &              &        &       &         \\
    \bottomrule
  \end{tabular}}
  \caption{Lean lines of code (LOC) and recorded source pages by
    section. ``Local'' denotes project LOC and ``Library'' the relevant Mathlib
    dependency footprint.
    }
  \label{tab:section-cost}
  \vspace{0.5em}
  \includegraphics[width=\linewidth]{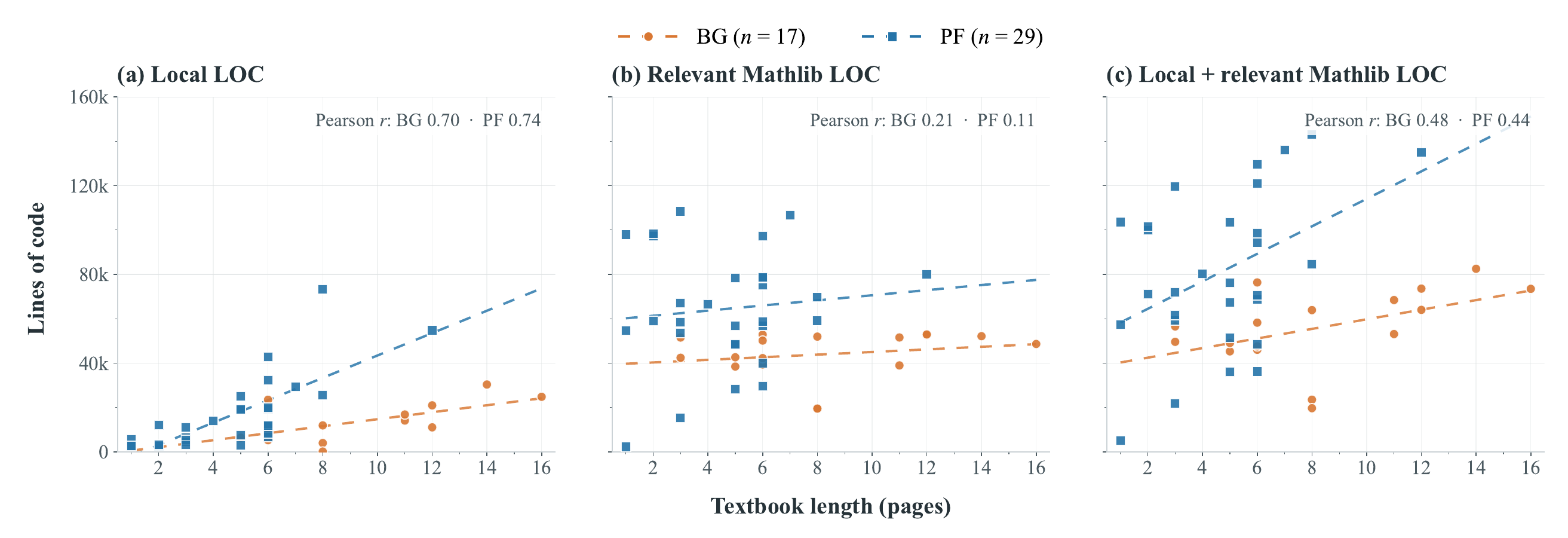}
  \captionof{figure}{Lean lines of code (LOC) against the number of source pages for
    all 46 rows in Table~\ref{tab:section-cost}. The three panels separate local
    project code, the relevant Mathlib dependency footprint, and their sum.}
  \label{fig:line-and-page}
\end{table}

Table~\ref{tab:section-cost} and Figure~\ref{fig:line-and-page} show that page count has a conditional, rather than universal, relationship with formalization size. Across all 46 sections, source pages and local Lean lines of code (LOC) have Pearson correlation $r=0.47$. Within each principal source family the association is stronger ($r=0.70$ for BG and $r=0.74$ for PF), yet the two families occupy visibly different vertical ranges in the left panel. Source identity therefore changes the baseline cost against which page length is informative. The eight-page rows make this concrete: BG Section~1 has 4,009 local LOC, whereas PF Section~9 has 73,314, more than eighteen times as many. Page count records how much space an author devotes to an argument, but not how much context is implicit, how much infrastructure the formal representation requires, or how finely the workflow decomposes the argument.

The tabular aggregates quantify the separation visible in the scatter plot. The BG rows contain 203,328 local LOC for 140 source pages, approximately 1,452 LOC per page. The PF rows, including Part~II and the appendices, contain 483,805 local LOC for 136 pages, approximately 3,557 LOC per page. Thus PF has about 2.45 times the observed formalization density of BG despite nearly equal page totals. This difference is not introduced by Peterfalvi Part~II: Part~I alone contains 369,722 local LOC for 86 pages, approximately 4,299 LOC per page. The comparison therefore points to a source-family effect, not simply to one unusually long auxiliary development.

The other two panels show why even this local-code relationship should not be interpreted as a predictor of total dependency footprint. Source pages have correlation $r=-0.12$ with relevant Mathlib LOC and only $r=0.12$ with local and relevant Mathlib LOC combined. Among the one-page PF sections, for example, Part~II Section~4.1 has 5,034 combined LOC, whereas Part~II Section~1.2 has 103,572---a factor of more than twenty. The relevant-Mathlib median is 48,636 LOC for BG and 59,113 for PF; the respective ranges are 19,528--52,931 and 2,291--108,413. A larger footprint can indicate broader reuse, heavier prerequisites, or both, and is not a direct count of missing lemmas. Taken together, the higher PF local density and much wider library footprint are consistent with the working hypothesis that BG draws on a more concentrated body of theory while PF crosses more varied interfaces. A direct audit of project-local foundational declarations would still be required to conclude that Mathlib support is intrinsically better for BG; proof style, agent decomposition, and source presentation remain potential confounders.

Section-level code volume measures the size of the formalization, but not the dependency structure that makes the task long-horizon. We therefore turn from source and code volume to the depth and connectivity of the declaration graph.

\subsubsection{Depth and Connectivity of the Declaration Dependency Graph}
\label{sec:analysis-dependency-graph}
\label{sec:analysis-long-horizon}

\paragraph{Graph definition and scope.} The formalized development induces a dependency graph over its Lean declarations. Here a declaration is a theorem or a definition. Let $G=(V,E)$ denote this graph, where each node in $V$ is a declaration and an edge $u\to v$ records that the declaration $u$ uses $v$. Thus, edges point from a declaration to one of its prerequisites. Lean's declaration discipline makes this dependency relation acyclic.

We analyze the induced subgraph whose nodes satisfy both of the following conditions: they are reachable from the declaration of the Bender--Suzuki theorem, and they belong to the formalization project rather than Mathlib. This subgraph contains 30,298 nodes and 186,187 edges.

\paragraph{Dependency depth and structural profile.} To quantify dependency depth, for each node $u\in V$ we compute the length of a longest directed path beginning at $u$:
\[
  d(u)=\max\{k:\text{there exists }u=u_0\to u_1\to\cdots\to u_k\}.
\]
The resulting distribution is shown in Figure~\ref{fig:maximum-path}. In the Bender--Suzuki project-only closure, the root theorem has depth 458. A corresponding root-specific computation for the Feit--Thompson theorem gives depth 430. These root depths place the final results at the end of exceptionally long chains of project declarations. For Challenge~\textsc{C1}, they quantify the long dependency horizon that must be maintained while composing the formalization.

\begin{figure}[t!]
  \centering
  \includegraphics[width=\linewidth]{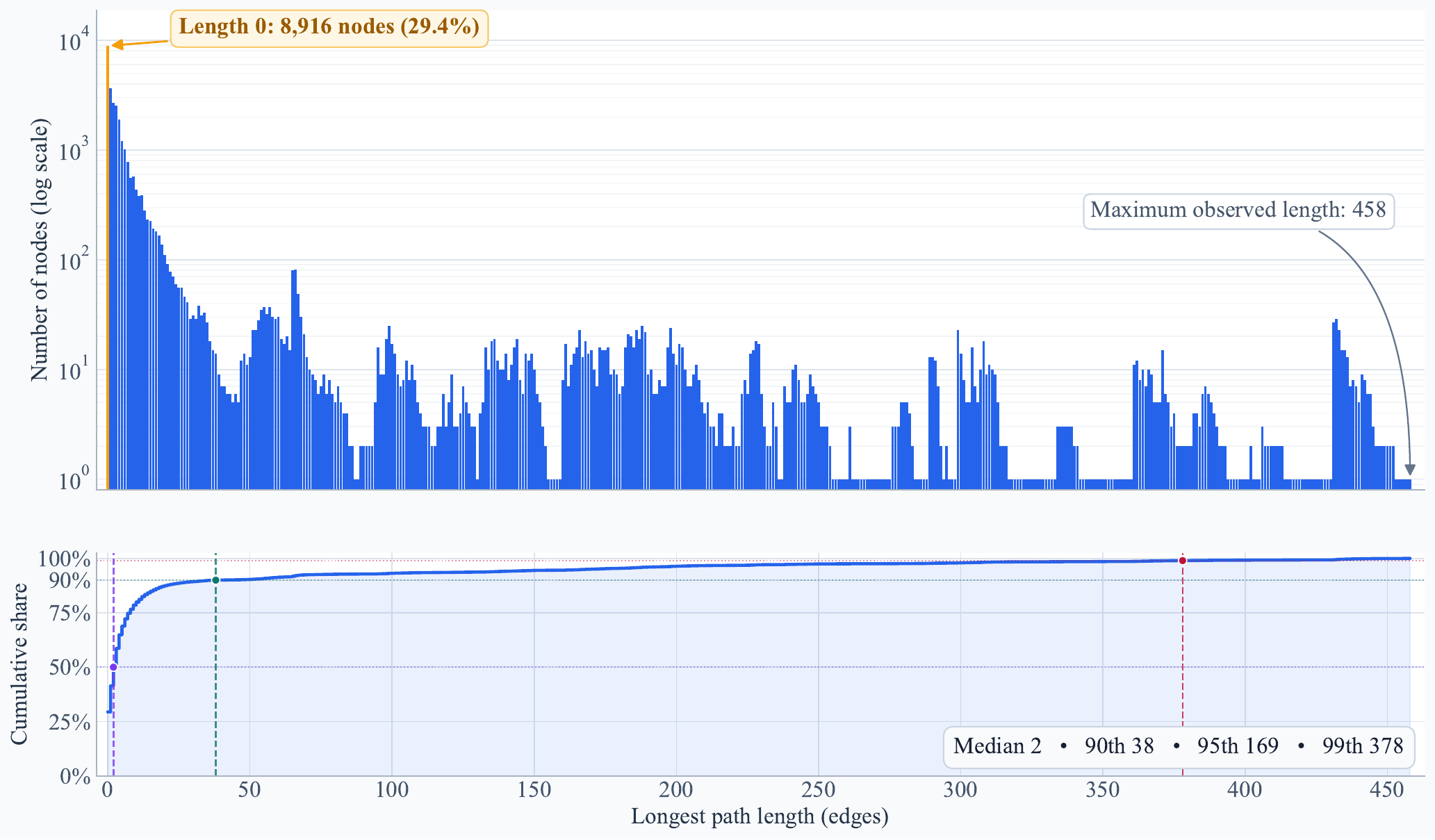}
  \caption{Distribution of the maximum path length starting from each project declaration in the Bender--Suzuki dependency closure.}
  \label{fig:maximum-path}
\end{figure}

The distribution is strongly skewed. Of the 30,298 project declarations, 8,916 (29.4\%) have depth zero, and the median depth is only two. Yet the 90th, 95th, and 99th percentiles are 38, 169, and 378, respectively. In particular, the depth increases by 209 between the 95th and 99th percentiles, even though half of the declarations have depth at most two. This contrast is a notable feature of the graph: the development is not uniformly deep, but has a broad shallow base together with a small, exceptionally deep backbone. The shape is more informative than the maximum alone. It localizes the long-horizon aspect of Challenge~\textsc{C1} to a relatively small upper tail of declarations that carries the long-range composition needed by the root theorem.

The histogram also contains narrow high-depth regions. In the interval $[342,360]$, for example, exactly one declaration occurs at each depth. The run records associate these declarations with the development around equations (13.6)--(13.12) of \cite{peterfalvi2000character}. The one-per-level pattern is consistent with a narrow logical bridge: at these depths the graph has little visible width, so progress toward the root may depend on a nearly serial piece of formal development.\footnote{The depth histogram alone, however, does not prove that the declarations form one adjacent path; that stronger statement requires direct inspection of the intervening edges.}

\paragraph{Reuse hubs and integration points.} For a declaration $v$, its \emph{in-degree} is the number of declarations that reference $v$, while its \emph{out-degree} is the number of declarations referenced by $v$. The five declarations with the largest values of each quantity are listed in Tables~\ref{tab:declaration-in-degree} and~\ref{tab:declaration-out-degree}.

\begin{table}[t!]
  \centering
  \small
  \begin{tabularx}{\linewidth}{
      >{\centering\arraybackslash}p{0.06\linewidth}
      >{\raggedright\arraybackslash}X
      >{\raggedright\arraybackslash}p{0.16\linewidth}
      >{\raggedleft\arraybackslash}p{0.12\linewidth}
      }
    \toprule
    Rank & Declaration                                   & Category   & In-degree \\
    \midrule
    1    & \path{Section1.ClassFunction}                 & Definition & 6,016     \\
    2    & \path{Section1.IsIrreducibleCharacterOnGroup} & Definition & 1,834     \\
    3    & \path{section9MaximalSubgroups}               & Definition & 1,792     \\
    4    & \path{ambientDerivedSubgroup}                 & Definition & 1,787     \\
    5    & \path{subgroupCentralizerIn}                  & Definition & 1,549     \\
    \bottomrule
  \end{tabularx}
  \caption{Top-five declarations by in-degree.}
  \label{tab:declaration-in-degree}
  \vspace{0.75em}
  \begin{tabularx}{\linewidth}{
      >{\centering\arraybackslash}p{0.06\linewidth}
      >{\raggedright\arraybackslash}X
      >{\raggedright\arraybackslash}p{0.16\linewidth}
      >{\raggedleft\arraybackslash}p{0.12\linewidth}
      }
    \toprule
    Rank & Declaration                                                                                  & Category & Out-degree \\
    \midrule
    1    & \path{BenderSuzuki.External.Higman.lemma12_distinct_pair_gap_classify_basic}                 & Theorem  & 536        \\
    2    & \path{FeitThompson.PFsection9.theorem_9_11_case_a_step_9_11_8_sourceData_of_noPairStep_sec9} & Theorem  & 351        \\
    3    & \path{BenderSuzuki.PFchapter3section1.theorem_c_of_Q1_ne_bot}                                & Theorem  & 239        \\
    4    & \path{Section12.theorem_12_17_lowerBoundData_source_leaf}                                    & Theorem  & 223        \\
    5    & \path{BenderSuzuki.External.Higman.lemma12_chain_typeBCD_with_isomorphic_criterion}          & Theorem  & 147        \\
    \bottomrule
  \end{tabularx}
  \caption{Top-five declarations by out-degree.}
  \label{tab:declaration-out-degree}
\end{table}

The in-degree ranking measures how widely a declaration is reused by the development. All five leading entries are definitions. Four expose recurring character-theoretic or group-theoretic interfaces: \path{Section1.ClassFunction}, \path{Section1.IsIrreducibleCharacterOnGroup}, \path{ambientDerivedSubgroup}, and \path{subgroupCentralizerIn}. The remaining entry, \path{section9MaximalSubgroups}, packages the maximal-subgroup structure used in the Feit--Thompson proof. Their concentration at the top of the ranking shows that a small set of definitions provides interfaces on which a large fraction of the development depends. The concentration is especially pronounced at the first rank: \path{Section1.ClassFunction} is referenced by 6,016 declarations, more than three times the in-degree of the second-ranked entry.

The out-degree ranking measures the breadth of the prerequisite context needed to establish a declaration. All five leading entries are theorems, and their large out-degrees identify proof or source-data packages that assemble many earlier results. These declarations are therefore integration points in the formal dependency structure, rather than the widely reused foundational interfaces identified by the in-degree ranking. The clean separation between the two top-five lists is not built into the degree definitions: high in-degree selects definitions that broadcast a shared vocabulary across the project, whereas high out-degree selects theorems that gather many prerequisites at difficult proof boundaries. Although the top five do not determine the full degree distribution, they expose two complementary forms of concentration that a long-horizon workflow must manage: stable reusable interfaces and localized integration points.

\paragraph{Project structure versus library substrate.} The project-only Bender--Suzuki closure contains 30,298 nodes and 186,187 edges. If Mathlib declarations are retained, the same root-specific closure contains 74,922 nodes and 1,444,499 edges. Including the library therefore multiplies the node count by about 2.5 and the edge count by about 7.8. Equivalently, the number of edges per node rises from about 6.1 in the project-only graph to about 19.3 in the library-inclusive graph. The disproportionate increase in edges is important: Mathlib contributes not merely additional prerequisite nodes, but a much more densely connected substrate. Much of the graph's raw connectivity therefore belongs to the library, whereas the project-only view isolates the structure introduced by this formalization. Reporting both views prevents a deep library dependency closure from being mistaken for an equally large body of newly formalized mathematics, while still recording the library structure on which the project relies.

\paragraph{Operational manifestation of the long dependency horizon.}
The declaration graph measures the static depth of the formalized theory, while the execution traces expose the corresponding operational horizon of the workflow. Across all recorded trajectories, the workflow consumed 64B tokens in total. The longest observed trajectory spanned 9.17 days, consumed 3.5 billion tokens, and underwent 606 context compactions. These measurements show that the largest formalization tasks extend far beyond a single model context or an uninterrupted proof attempt. Progress must instead survive repeated compression of the interaction history while preserving the current proof state, the relevant section-level context, and decisions accumulated over earlier stages. This provides an operational counterpart to the depth of the declaration DAG and motivates the graph-based state tracking mechanism for {\sc Prover} described in Section~\ref{sec:prover}.

\medskip
\noindent
Taken together, these measurements expose two complementary dimensions of Challenge~\textsc{C1}. Source discovery makes the boundary of the task open-ended, while the declaration dependency graph makes its internal horizon exceptionally deep but highly concentrated. The difficulty lies not simply in formalizing many source pages, but in discovering and composing the small set of long dependency chains that connect a broad base of local declarations to the root theorems. The degree rankings further suggest where this composition is organized: reusable definitions stabilize shared interfaces, while a small number of theorems integrate unusually broad prerequisite contexts.

\subsection{Source Defects and Misalignments Identified during the Formalization Process}
\label{sec:issues}

In this section, we present representative mathematical issues identified in the
literature during our project and classify each issue according to the source
defects and misalignments defined in
Section~\ref{sec:source-defects-and-misalignments}. Most of these issues were resolved automatically by
the workflow, either through the {\sc Prover} procedure or by
{\sc ReconcilerAgent}. Because such automatically resolved issues are recorded only
in execution logs and are not reported separately by the workflow, our retrospective
review may not have recovered every instance. A smaller number of issues could not
be resolved automatically and were escalated for human investigation; we report all
such cases encountered in the project.

The first two examples arose while formalizing the proof of the Odd Order Theorem
based on \cite{bender1994local,peterfalvi2000character}. The existing Rocq
formalization~\cite{gonthier2013machine} may also have encountered these issues.
Our workflow nevertheless identified them independently. Moreover, for one of these
issues, namely the incompatibility between the definitions of type~I maximal
subgroup, we adopted a different resolution, as described below.

\paragraph{Incompatibility between definitions of type I maximal subgroups \textnormal{\textsc{[Mal1]}}.} The proof of the Odd Order Theorem proceeds by choosing a minimal
counterexample \(G\) and analyzing the structure of its maximal subgroups.
These maximal subgroups are divided into five classes, called types~I--V.
However, \cite{bender1994local} and \cite{peterfalvi2000character} give
different definitions of type~I maximal subgroups. Condition~(Iii) in
\cite[Section~16]{bender1994local} requires that
\begin{quote}
each complement $E$ to $H$ in $M$ contains a normal abelian subgroup $A$ such that $C_E(x) \subset A$ for all $x \in H^{\#}$,
\end{quote}
whereas \cite[Definition~8.1]{peterfalvi2000character} requires only that
\begin{quote}
    there is a complement $U$ of $H$ containing a normal abelian subgroup $U_1$ such that $C_U(x) \subset U_1$ for all $x\in H^{\#}$.
\end{quote}
Thus, the definition in \cite{bender1994local} is formally stronger than
that in \cite{peterfalvi2000character}: the former quantifies over all
complements, whereas the latter requires only the existence of one such
complement. Nevertheless, the two conditions are equivalent by the
Schur--Zassenhaus theorem, since complements to \(H\) in \(M\) exist and
are conjugate.

This case was resolved automatically by the {\sc Prover} procedure and was identified through a retrospective examination of the execution logs. In particular, the {\sc Prover} procedure resolves this
incompatibility by establishing a lemma that bridges the two definitions.
By contrast, the existing Rocq formalization of the Odd Order
Theorem~\cite{gonthier2013machine} adopts the latter condition as its
definition of type~I maximal subgroups and formalizes the results from
\cite{bender1994local} using the formulation in
\cite{peterfalvi2000character}.

\paragraph{Typographical errors in the definitions of type~I maximal subgroups \textnormal{\textsc{[Def1]}}.}
Condition~(Iv)(c) in the definition of type~I maximal subgroups in \cite{bender1994local} states
\begin{quote}
    for every $p \in \pi(H)$, $p\in \pi^*$ and the exponent of $M/H$ divides $p-1$; for some such prime $p$, we have $\mathcal{O}_{p'}(M)$ is cyclic.
\end{quote}
While attempting to prove Proposition~16.1 of \cite{bender1994local} under
this definition, the autoformalization workflow stalled. After repeated
attempts, {\sc Prover} returned \(\mathsf{Failure}\) to {\sc FormaTheoria},
triggering an escalation for human investigation. Manual examination of the
proof of Proposition~16.1 revealed that the final condition should instead read:
\begin{quote}
    for every $p \in \pi(H)$, $p\in \pi^*$ and the exponent of $M/H$ divides $p-1$; for some such prime $p$, we have $\mathcal{O}_{p'}(H)$ is cyclic.
\end{quote}
Thus, \(\mathcal{O}_{p'}(M)\) in the stated definition appears to be a
typographical error for \(\mathcal{O}_{p'}(H)\). The corresponding definition
in \cite[Definition~8.3]{peterfalvi2000character} contains the same error.

\paragraph{Missing condition in Peterfalvi \textnormal{\textsc{[Def1]}}.}
The proof of Appendix~IV, Lemma~2(c) in
\cite{peterfalvi2000character} requires the additional condition that
\(\lvert Q_1\rvert\) is odd, but this condition is omitted from the lemma
statement. In Peterfalvi's application of the lemma, however,
\(\lvert Q_1\rvert\) is already known to be odd from the earlier decomposition
in \cite[Part~II, Chapter~1]{peterfalvi2000character}. The application can therefore be justified by supplying this hypothesis implicitly in that context. The condition must nevertheless be added to the lemma statement for the lemma itself to hold. This issue was resolved automatically by {\sc ReconcilerAgent} and identified through a
retrospective examination of the execution logs.

\paragraph{Error in Theorem 8.27 in Huppert I \textnormal{\textsc{[Def1]}}.}
The statement of Dickson's theorem in
\cite[Satz~II.8.27]{huppert1967endliche} includes \begin{quote}
semidirect products of elementary abelian groups of order $p^m$ with cyclic groups of order $t$, where $t \mid \frac{p^m-1}{d}$ and $t \mid p^f - 1$.
\end{quote}
The {\sc Prover} procedure found counterexamples to this statement and
escalated the issue for human investigation. Manual examination determined
that the correct statement is
\begin{quote}
semidirect products of elementary abelian groups of order $p^m$ with cyclic groups of order $t$, where $t \mid p^m-1$ and $t \mid \frac{p^f - 1}{d}$.
\end{quote}
Thus, the factor \(d\) appears in the wrong divisibility condition in the
stated theorem.

\paragraph{Error in Lemma~11 of Higman \textnormal{\textsc{[Def2]}}.}
In the proof of Lemma~11 in \cite[p.~88]{higman1963suzuki}, the basis
\(u_0,\dots,u_m\) should instead be \(u_0,\dots,u_{m-1}\). This indexing
error was corrected automatically by the {\sc Prover} procedure and identified through a
retrospective examination of the execution logs.

\paragraph{Inconsistent definitions of Suzuki 2-groups \textnormal{\textsc{[Mal1]}}.}
In the definition of Suzuki 2-groups, Peterfalvi \cite{peterfalvi2000character} requires the action to be regular:
\begin{quote}
    A Suzuki 2-group is a 2-group $P$ such that $P$ is non-abelian, $P$ has at least two involutions and there is a cyclic group $K$ which acts faithfully on $P$ and regularly on the set of involutions of $P$.
\end{quote}
In contrast, Higman \cite{higman1963suzuki} only assumes that the action is transitive. The {\sc Prover} automatically derives transitivity from regularity and therefore can invoke the relevant results in \cite{higman1963suzuki}. Note that Higman later proved the non-trivial result that transitivity in this setting actually implies regularity; however, this theorem has not been formalized by the {\sc Prover}. This issue was identified through a retrospective examination of the execution logs.

\subsection{Examples of Automatically Fixed Translation Mistakes}
\label{sec:examples-translation-mistakes}

The initial translation of a source item may fail to preserve its intended mathematical content, particularly when the source relies on implicit context or when hypotheses and conclusions are incorrectly organized in the formal statement. Such errors may be exposed either by a failed proof attempt or by the failure of a proved theorem to support a downstream argument. In these cases, the {\sc ReconcilerAgent} compares the translation with the relevant source context and downstream uses, identifies the discrepancy, and repairs the formal statement. The following examples illustrate two translation mistakes detected and corrected through this process.

\paragraph{Implicit source context omitted in translation
\textnormal{\textsc{[InT1]}}.}
Part~I of \cite{peterfalvi2000character} is written in a style that leaves many standing assumptions implicit. For example, (13.18) states:
\begin{quote}
    Let $j$ be such that $0 < j < p$ and let $\beta_j = \operatorname{Ind}_{PW_1}^S 1_{P W_1} - \mu_{0j}$.
    (a) $\mathrm{supp}(\beta_j) \subset P^{\#} \cup (W - (W_1\cup W_2))^S \subset A_0(S)$.
    (b) $\|\beta_j\|^2 = \frac{u-1}{q} + 2$.
    (c) $\Gamma = \beta_j^{\tau} - 1_G + \eta_{0j}$ is independent of $j$, orthogonal to $1_G$ and real.
    (d) Set $\Gamma = X+Y$ where $X$ is a linear combination of the functions $\eta_{ik}$ and $Y$ is orthogonal to the functions $\eta_{ik}$. Then $\|Y\|^2 \le \frac{u-1}{q}$.
\end{quote}
Here, $G$ is a chosen minimal counterexample to the Odd Order Theorem, and $S\leq G$ is a particular subgroup selected in (13.1) on Page~75, nine pages before (13.18) on Page~84. Similarly, $\mu$ is a class function introduced in (13.1), which in turn refers back to (4.3) for its construction and relevant properties. The initial translation of (13.18) omitted the context associated with these particular choices of $S$ and $\mu$. Consequently, the {\sc Prover} could not establish the translated statement because the formal context lacked properties of $S$ and $\mu$ required by the proof. The {\sc ReconcilerAgent} repaired the translation by comparing (13.18) with (13.1) and making the required context explicit.

\paragraph{Translation weakened theorem statement
 \textnormal{\textsc{[InT2]}}.}
The initial translation of Peterfalvi's (12.6)~\cite{peterfalvi2000character} was incorrect: the conclusions of \texttt{theorem\_12\_6\_statement} were mistakenly included in \\ \texttt{theorem\_12\_6\_source\_data}, which appears as an assumption of \texttt{theorem\_12\_6\_statement}. The resulting theorem therefore admitted a trivial proof but was unusable in downstream arguments, such as Peterfalvi's (12.16) and (14.2), because constructing \texttt{theorem\_12\_6\_source\_data} already required establishing the intended conclusions. Accordingly, the {\sc Prover} initially proved \texttt{theorem\_12\_6} trivially but later failed to apply it while proving Peterfalvi's (12.16). This failure triggered the {\sc ReconcilerAgent}, which corrected the statement of \texttt{theorem\_12\_6} by comparing it with its use in (12.16).

\subsection{Reconciliation Hotspots and Repair Boundaries}
\label{sec:analysis-reconciliation}

Reconciliation addresses theory-integration failures arising from cross-source misalignment (Challenge \textsc{C2}), incorrect Lean translation (Challenge~\textsc{C3}), and source defects (Challenge~\textsc{C4}). Such failures typically surface as an unprovable translation or downstream application conflict of type~\textnormal{\textsc{[InT1]}} or~\textnormal{\textsc{[InT2]}}, although diagnosis may reveal an underlying~\textnormal{\textsc{[Mal1]}}, \textnormal{\textsc{[Mal2]}}, \textnormal{\textsc{[Def1]}}, or~\textnormal{\textsc{[Def2]}} issue, as illustrated in Section~\ref{sec:issues}. Table~\ref{tab:reconciler-activity} reports the exact modification counts and declaration populations, while Figure~\ref{fig:reconciler-modification} normalizes these counts by source-group size and orders the groups by modification share. The normalized view distinguishes absolute modification volume from the relative extent of the affected code surface. Because one reconciliation may modify several declarations, whereas an unresolved case may be escalated without modifying any, these measurements describe the affected code surface rather than agent invocations, distinct incompatibilities, successful repairs, or error rates.

\begin{table}[p]
  \centering
  \scriptsize
  \setlength{\tabcolsep}{3pt}
  \renewcommand{\arraystretch}{0.98}
  \resizebox{\linewidth}{!}{%
  \begin{tabular}{lrr@{\qquad}|@{\qquad}lrr}
    \toprule
    Source Group            & \makecell{Declarations\\Modified} & \makecell{Total\\Declarations} &
    Source Group            & \makecell{Declarations\\Modified} & \makecell{Total\\Declarations} \\
    \midrule
    BG~Section~5            & 4                     & 127                & PF~Part I~Section~14    & 21                    & 722                \\
    BG~Section~6            & 1                     & 114                & PF~Part II~Chapter~1.1  & 2                     & 209                \\
    BG~Section~13           & 2                     & 259                & PF~Part II~Chapter~1.2  & 1                     & 86                 \\
    BG~Section~16           & 1                     & 356                & PF~Part II~Chapter~1.3  & 1                     & 130                \\
    PF~Part I~Section~6     & 1                     & 711                & PF~Part II~Chapter~2    & 1                     & 434                \\
    PF~Part I~Section~7     & 4                     & 378                & PF~Part II~Chapter~4.2  & 2                     & 64                 \\
    PF~Part I~Section~8     & 8                     & 537                & PF~Part II~Appendix~III & 1                     & 47                 \\
    PF~Part I~Section~9     & 27                    & 1,425              & PF~Part II~Appendix~IV  & 3                     & 200                \\
    PF~Part I~Section~10    & 82                    & 914                & Huppert~Chapter~II      & 2                     & 120                \\
    PF~Part I~Section~11    & 1                     & 467                & Huppert~Chapter~IV      & 1                     & 268                \\
    PF~Part I~Section~12    & 1                     & 424                & Huppert~Chapter~XI      & 5                     & 697                \\
    PF~Part I~Section~13    & 110                   & 1,294              & Higman                  & 2                     & 434                \\
    \midrule
    \multicolumn{3}{r}{} & Total                   & 284                   & 10,417             \\
    \bottomrule
  \end{tabular}}
  \caption{Declarations modified by {\sc ReconcilerAgent} relative to the total declaration count in each source group.}
  \label{tab:reconciler-activity}
  \vspace{0.5em}
  \includegraphics[width=\linewidth]{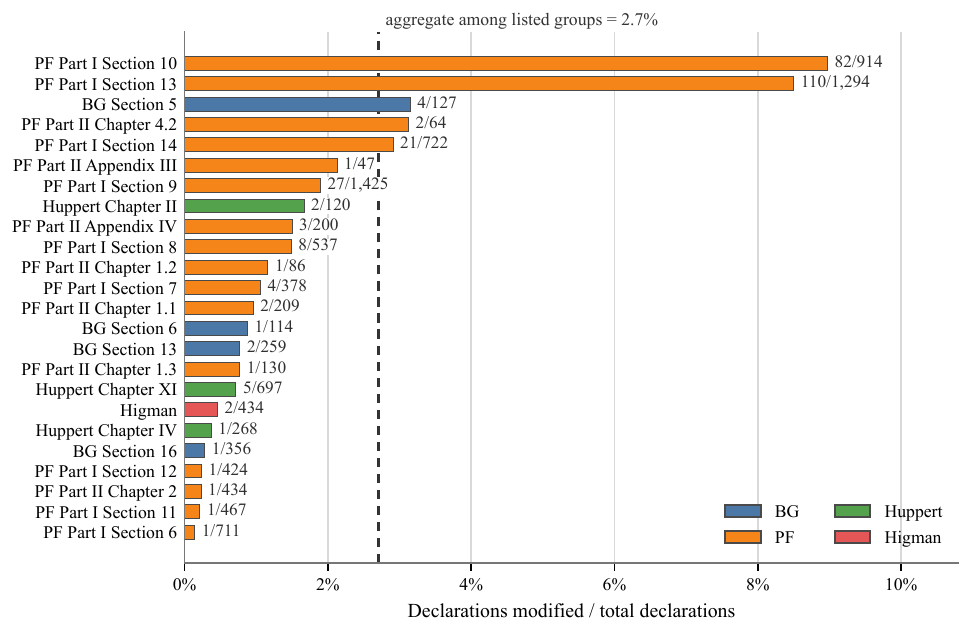}
  \captionof{figure}{Source groups ordered by the share of declarations modified by {\sc ReconcilerAgent}. Bar lengths show the number of modified declarations divided by the total declaration count in the corresponding source group; labels report the underlying counts, and colors distinguish source families. The dashed line marks the aggregate modification share across the source groups listed in Table~\ref{tab:reconciler-activity}.}
  \label{fig:reconciler-modification}
\end{table}

\paragraph{Where reconciliation concentrates.}  Table~\ref{tab:reconciler-activity} reports the number of declarations modified during reconciliation together with the total declaration population of each audited source group, while Figure~\ref{fig:reconciler-modification} normalizes these counts by source-group size. These quantities measure the formal code surface affected by reconciliation. They do not directly measure the number of reconciliation invocations, distinct incompatibilities, or source errors, because a single incompatibility may require coordinated changes to several declarations.

The modifications are highly concentrated. Across the source groups represented in Table~\ref{tab:reconciler-activity}, 284 of 10,417 declarations were modified, giving an aggregate share of \(2.7\%\). PF Part~I accounts for 255 of these 284 modifications. Within PF Part~I, Sections~9, 10, and~13 alone account for 219 modifications, or \(77.1\%\) of all modifications in the table, although they contain only 3,633 of the 10,417 declarations. The concentration remains visible after normalization. Sections~10 and~13 have modification shares of \(8.97\%\) and \(8.50\%\), respectively, and together contain 192 modified declarations among 2,208 declarations, for a combined share of \(8.7\%\). More broadly, Sections~9--14 have a modification share of \(242/5{,}246=4.6\%\), compared with \(13/1{,}626=0.8\%\) for Sections~6--8. Thus, the later PF Part~I sections exhibit a modification share approximately \(5.8\) times as large.

Section size alone does not account for this pattern. For example, Section~9 contains 1,425 declarations, more than either Section~10 or Section~13, but its modification share is only \(1.89\%\). To explore whether the concentration is instead associated with cross-section integration, we analyze the declaration-dependency graph. Let \(V_s\) denote the declarations owned by Section~\(s\). Because an edge \(u\to v\) indicates that \(u\) depends on \(v\), define
\[
d_s^+(u)
= \lvert\{v : u\to v,\ \operatorname{owner}(v)\neq s\}\rvert,
\qquad
C_s
= \sum_{u\in V_s} d_s^+(u).
\]
Here, \(C_s\) counts the dependency edges from declarations in Section~\(s\) to declarations owned by other sections. It therefore measures the extent to which the formalization of Section~\(s\) relies on interfaces with previously formalized material.

Table~\ref{tab:reconciler-boundary-metrics} provides supporting evidence for an association between such cross-section dependencies and reconciliation concentration. Sections~10 and~13 have the two largest cross-section dependency counts, \(10{,}203\) and \(11{,}631\), as well as the two highest modification shares. Across the fourteen PF Part~I sections, the modification share has Pearson correlation \(r=0.861\) and Spearman correlation \(\rho=0.820\) with \(C_s\), both with \(p<0.001\). Its association with section size is weaker: the corresponding correlations are \(r=0.607\) (\(p=0.021\)) and \(\rho=0.483\) (\(p=0.080\)). 
These exploratory comparisons suggest that the observed concentration is more closely associated with cross-section integration than with the number of declarations alone.

\begin{table}[t]
  \centering
  \begin{tabular}{lrrr}
    \toprule
    PF Section & $|V_s|$ & $C_s$ & \makecell{{\sc ReconcilerAgent}\\modification share} \\
    \midrule
    1  &   737 &    106  &  0\%    \\
    2  &   329 &   163   &   0\%    \\
    3  &   682 &   1,715 &   0\%    \\
    4  &   432 &  1,767  &   0\%    \\
    5  &   619 &   2,967 &   0\%    \\
    6  &   711 &   4,070 &   0.14\% \\
    7  &   378 &   1,038 &  1.06\% \\
    8  &   537 &   4,026 &   1.49\% \\
    9  & 1,425 &   6,308 &   1.89\% \\
    10 &   914 &  10,203 &   8.97\% \\
    11 &   467 &  5,365  &   0.21\% \\
    12 &   424 &   3,942 &   0.24\% \\
    13 & 1,294 &  11,631 &   8.50\% \\
    14 &   722 &  8,442  &   2.91\% \\
    \bottomrule
  \end{tabular}
  \caption{Cross-section dependency counts for PF Part I sections.}
  \label{tab:reconciler-boundary-metrics}
\end{table}

Inspection of the corresponding Lean development further suggests a plausible structural explanation for this association. In PF Part~I Section~10, the file \texttt{PFsection10/Basic.lean} organizes source hypotheses and their carrier data through \texttt{section10FullFourSixSourceNotation},\\ \texttt{section10FourSixSupportedPackage}, and \texttt{section10SupportedFourSixData}. The first declaration records the source-notation layer, whereas the latter two package the supported assumptions and carrier data. The file \texttt{PFsection10/PFsection10\_11.lean} expands these packages into families of source-data and supported-data declarations, together with bridge lemmas that project fields from bundled hypotheses or transport existing results into the supported representation. For example, the odd-order, vanishing, and parity fields of Hypothesis~(10.4) are extracted through separate theorem interfaces whose proof terms directly select the corresponding fields. PF Part~I Section~13 exhibits the same pattern across a larger module boundary: its 21 source files combine material from Sections~3, 5, 6, 8, 9, 10, and~11. In \texttt{PFsection13/PFsection13\_2.lean}, separate source, rank-choice, and branch-data predicates are connected by conversion lemmas, while \texttt{PFsection13/PFsection13\_Common.lean} stages the numerical argument of Theorem~(13.10) through raw-source, total-norm, cardinal-formula, and strict-comparison data, with bridge lemmas between successive representations. These intermediate declarations generally expose source contexts or transport existing results rather than express independent mathematical conclusions. Consequently, a mismatch at one representation boundary may require coordinated revisions to several packages, projections, and transport lemmas. Thus, integration across formal interfaces may expand the reconciliation footprint of a single conceptual incompatibility, requiring coordinated modifications to multiple packages, projections, conversion lemmas, and transport lemmas.

\paragraph{Why local reconciliation can succeed.} Many conflicts are repairable because the two sides already contain enough source evidence to justify a mathematical bridge. For a type~\textnormal{\textsc{[Mal1]}} mismatch, the definitions may be mathematically equivalent even though Lean does not identify them definitionally; an equivalence lemma or conversion construction can make the identification explicit. For a conflict manifesting as~\textnormal{\textsc{[InT2]}}, the current candidate can often be revised to use the approved dependency through an adapter, without changing either source-level statement. The package and projection layers in PF Sections~10 and~13 make these boundaries verbose, but they also localize the required transports and provide named intermediate invariants against which a repair can be checked.

The approval rule is central to this capability. {\sc ReconcilerAgent} may modify the current candidate, add justified compatibility declarations, replace the proof term of an approved theorem while preserving its type, or revise explicitly authorized theorem statements. It may not silently weaken an un-authorized approved hypothesis, strengthen an un-authorized approved conclusion, or identify two objects merely because doing so makes elaboration succeed. The reconciliation is accepted only if the sources support the bridge and all existing downstream declarations continue to elaborate. In this sense, reconciliation is constrained theory integration rather than unrestricted refactoring.

Section~\ref{sec:issues} provides concrete examples of this boundary. The two formulations of type~I maximal subgroups constitute a type~\textnormal{\textsc{[Mal1]}} misalignment: a source-supported equivalence bridge is possible because Schur--Zassenhaus supplies the missing mathematical identification. The omitted odd-order condition in Peterfalvi is a type~\textnormal{\textsc{[Def1]}} defect that {\sc ReconcilerAgent} could repair because the required hypothesis was already established in the surrounding source context. Likewise, the regular formulation of Suzuki 2-groups is sufficient for using a result stated under transitivity, because the needed implication has a supported direction. These cases are repairable not because reconciliation tolerates inconsistency, but because the available sources establish a local path from the stronger or differently packaged context to the interface required downstream.

\paragraph{Limits and escalation boundary.} Reconciliation cannot manufacture mathematical evidence that is absent from the sources. A type~\textnormal{\textsc{[Mal2]}} conflict may involve genuinely distinct objects sharing a name; identifying them would be unsound, while separating them may require a nonlocal redesign of the formal theory. A type~\textnormal{\textsc{[Def1]}} or~\textnormal{\textsc{[Def2]}} defect may require deciding what the author intended, supplying a missing theorem, or making unauthorized correction to an approved theorem statement. Similarly, an~\textnormal{\textsc{[InT1]}} error in the type of an already approved declaration cannot be repaired merely by replacing its proof term. These cases exceed the local authority of {\sc ReconcilerAgent} and must be escalated or returned to an upstream translation and approval step.

The examples in Section~\ref{sec:issues} again show the distinction. The misplaced divisibility factor in Huppert's statement required human examination, whereas the indexing error in Higman's proof could be corrected during proof construction without changing the theorem statement. For Suzuki 2-groups, the workflow uses the supported implication from regularity to transitivity, but it does not claim the unformalized converse. The operative criterion is therefore not whether a conflict is small syntactically, but whether a local repair is justified by available evidence, respects approved interfaces, and preserves every downstream use.

Finally, the aggregate does not measure the success rate or cost of reconciliation. A high modification share may reflect one broad interface repair rather than many independent failures, and a difficult unsuccessful reconciliation may leave no modified declaration at all. This distinction is operationally important because reconciliation is globally exclusive in the parallel workflow of Section~\ref{sec:batch-parallelization}; repeated calls at a hotspot may pause otherwise independent workers, but declaration counts cannot quantify that delay. Establishing operational effectiveness would require call-level traces recording the diagnosed issue type, attempted repair, review outcome, human escalation, downstream declarations rechecked, and time spent at the reconciliation barrier. The present data instead support a structural conclusion: reconciliation work is concentrated where source boundaries are exposed as formal interfaces, and automatic repair is effective only when those interfaces can be connected by a local, source-supported construction.

\subsection{Review as an Iterative Quality Gate}
\label{sec:analysis-review-reconcile}

\begin{table}[t!]
  \centering
  \small
  \begin{tabular}{lc@{\qquad\qquad}lc}
    \toprule
    Section & Review Rejects & Section & Review Rejects \\
    \midrule
    PF Sec.~1 & 3 & PF Sec.~8  & 1 \\
    PF Sec.~2 & 0 & PF Sec.~9  & 1 \\
    PF Sec.~3 & 1 & PF Sec.~10 & 0 \\
    PF Sec.~4 & 2 & PF Sec.~11 & 2 \\
    PF Sec.~5 & 2 & PF Sec.~12 & 1 \\
    PF Sec.~6 & 1 & PF Sec.~13 & 1 \\
    PF Sec.~7 & 0 & PF Sec.~14 & 1 \\
    \midrule
    \multicolumn{3}{r}{Total} & 16 \\
    \bottomrule
  \end{tabular}
  \caption{Rejected review outcomes for translation work in PF Part~I.}
  \label{tab:review-count}
\end{table}

Each section-level translation was reviewed repeatedly until it was accepted. Consequently, a section with \(r\) rejected outcomes required \(r+1\) review rounds. As shown in Table~\ref{tab:review-count}, the fourteen PF Part~I sections generated 16 rejected outcomes and therefore required 30 review rounds in total. Only three sections were accepted on the first review; the other eleven, or \(78.6\%\), were returned for revision at least once. Seven sections were accepted in the second round, three in the third, and one in the fourth.

These results provide direct operational evidence that the review agent affected the translation process. For most sections, it did not simply confirm the first candidate: it identified sufficient problems to prevent that candidate from advancing past the review stage and triggered another translation attempt. At the same time, the varying number of rounds shows that review did not impose a fixed rejection cycle. Some translations passed immediately, whereas others required up to three revisions. All sections were eventually accepted within four rounds, suggesting that the feedback was actionable enough for the translation process to converge rather than repeatedly producing unresolved candidates.

The data establish the usefulness of review as an active quality gate, but not its semantic accuracy in isolation. A rejected outcome may contain one or several findings, and repeated rejections may concern the same underlying issue. Moreover, eventual acceptance does not by itself prove that every remaining translation choice is correct. A stronger evaluation would require issue-level review records and an independent assessment of accepted candidates. The present evidence nevertheless shows that the review agent made consequential decisions: it intervened in most section-level translations and converted a single-pass generation process into an iterative generate--check--revise workflow.

\subsection{Comparison with Existing Rocq Formalization}
\label{sec:analysis-rocq}

A machine-checked proof of the Odd Order theorem, constructed in the Coq Proof Assistant using the Mathematical Components library, was announced in September 2012 by Gonthier and collaborators \cite{gonthier2013machine}. The Coq project was renamed The Rocq Prover in 2025 \cite{rocq2025release}; we use the current name below. At repository level, our Lean formalization contains substantially more lines of code than the Rocq formalization. This is a descriptive difference in code volume. It does not, by itself, measure development effort, mathematical difficulty, or proof quality. The comparison here is deliberately restricted to the Odd Order development. At the source level, that proof is organized around BG and PF Part~I; PF Part~II supports the Bender--Suzuki development and is not used by the Rocq Odd Order development. Accordingly, Table~\ref{tab:rocq} compares the shared BG and PF Part~I section breakdown, while the root-closure totals below additionally include every project declaration---that is, every named theorem or definition---needed by the root theorem, including declarations outside that breakdown. Neither quantity represents all Lean code in our repository. We therefore compare the two developments along three dimensions: their organization and proof routes over the shared sources, their code volume, and their checks on the final theorem and its logical assumptions.

\paragraph{Shared source structure, different proof routes.}
Both developments replace the original 255-page proof of Feit and Thompson \cite{feit1963solvability} with the same later two-volume route: the BG volume on local analysis \cite{bender1994local} and PF Part~I on character theory \cite{peterfalvi2000character}. Both also preserve the section structure of those sources. The Rocq development consists of the files \texttt{BGsection1.v} to \texttt{BGsection16.v}, then \texttt{BGappendixAB.v} and \texttt{BGappendixC.v}, then \texttt{PFsection1.v} to \texttt{PFsection14.v}. Our directories carry the same names, with no counterpart to \texttt{BGappendixAB.v}. A section-by-section comparison is therefore possible, and we report one in Table~\ref{tab:rocq}. The unusually close alignment of source boundaries makes this comparison more informative than a repository-wide line count, while the remaining differences in libraries and proof style still preclude a direct effort comparison.

The shared decomposition comes from the shared sources, but the developments diverge in four aspects:

\begin{enumerate}
\item \textbf{Proof route through BG's appendices.} For the principal proof-route divergence, our development departs from BG. BG's Appendix~B proves Puig's theorem, presenting it as the analogue of \cite[Theorem~8.2.11]{gorenstein1968finite} on $Z(J(S))$, and the Rocq development follows this route: it formalizes BG's Appendices~A and~B in \texttt{BGappendixAB.v} and uses the Puig subgroup in four further files. We instead prove Glauberman's theorem on $Z(J(S))$ itself, following \cite[Chapter~8]{gorenstein1968finite}. No subgroup of Puig's kind occurs under any name in our formalization: every subgroup we define as a supremum of a family of subgroups is a $p$-core, a $\pi$-core, the Fitting subgroup, or Thompson's $J$.
\item \textbf{Use of analysis.} The Rocq proof is free of analysis by construction and uses only the field \texttt{algC} of algebraic numbers containing character values. Our development, by contrast, uses both the complex numbers of Mathlib and real analysis.
\item \textbf{Division into declarations.} The declaration counts are likewise not in section-by-section correspondence: our development has fewer declarations than the Rocq development in five sections, but eighty times as many in PF Part~I Section~6. Such a spread is consistent with an autoformalized development in which many additional declarations are helper lemmas and their number depends on the decomposition found by {\sc RefineAgent}, rather than directly on the source structure.
\item \textbf{Type-I maximal subgroups.} The developments resolve the discrepancy between the definitions of type~I maximal subgroups discussed in Section~\ref{sec:issues} in opposite ways. The Rocq formalization adopts the existential formulation in \cite[Definition~8.3]{peterfalvi2000character}, whereas ours adopts the universal formulation in \cite[Section~16]{bender1994local}. Because the two formulations are equivalent, this representational choice does not affect the downstream mathematics.
\end{enumerate}

\begin{table}[t!]
  \centering
  \scriptsize
  \setlength{\tabcolsep}{6pt}
  \renewcommand{\arraystretch}{1.08}
  \resizebox{0.95\textwidth}{!}{%
  \begin{tabular}{c>{\hspace{-4pt}}r<{\hspace{4pt}}>{\hspace{-4pt}}r<{\hspace{4pt}}r@{\qquad}|@{\qquad}c>{\hspace{-4pt}}r<{\hspace{4pt}}>{\hspace{-4pt}}r<{\hspace{4pt}}r}
    \toprule
    \makecell{BG\\Section} & \multicolumn{1}{c}{\makecell{Lean\\Lines of Code}} & \multicolumn{1}{c}{\makecell{Rocq\\Lines of Code}} & Ratio &
    \makecell{PF Part I\\Section} & \multicolumn{1}{c}{\makecell{Lean\\Lines of Code}} & \multicolumn{1}{c}{\makecell{Rocq\\Lines of Code}} & Ratio \\
    \midrule
    1          & 4,009   & 1,334  & 3.0   & 1          & 12,236  & 755     & 16.2  \\
    2          & 174     & 1,154  & 0.2   & 2          & 7,639   & 830     & 9.2   \\
    3          & 24,842  & 1,829  & 13.6  & 3          & 19,287  & 1,874   & 10.3  \\
    4          & 14,140  & 1,411  & 10.0  & 4          & 11,882  & 995     & 11.9  \\
    5          & 6,838   & 531    & 12.9  & 5          & 19,119  & 1,611   & 11.9  \\
    6          & 6,605   & 317    & 20.8  & 6          & 25,513  & 1,347   & 18.9  \\
    7          & 5,287   & 972    & 5.4   & 7          & 6,648   & 835     & 8.0   \\
    8          & 7,176   & 398    & 18.0  & 8          & 19,291  & 1,133   & 17.0  \\
    9          & 6,425   & 470    & 13.7  & 9          & 73,314  & 2,230   & 32.9  \\
    10         & 16,858  & 1,500  & 11.2  & 10         & 42,717  & 1,225   & 34.9  \\
    11         & 5,005   & 440    & 11.4  & 11         & 24,995  & 1,200   & 20.8  \\
    12         & 30,374  & 2,680  & 11.3  & 12         & 19,884  & 1,363   & 14.6  \\
    13         & 11,952  & 1,119  & 10.7  & 13         & 54,809  & 2,190   & 25.0  \\
    14         & 20,985  & 2,511  & 8.4   & 14         & 32,388  & 1,261   & 25.7  \\
    15         & 23,516  & 1,508  & 15.6  &            &          &         &       \\
    16         & 11,074  & 1,362  & 8.1   &            &          &         &       \\
    Appendix~C & 8,068   & 778    & 10.4  &            &          &         &       \\
    \midrule
    Total      & 203,328 & 20,314 & 10.0  & Total      & 369,722 & 18,849 & 19.6  \\
    \bottomrule
  \end{tabular}
  }
  \caption{Lines of code by source section in the Lean and Rocq formalizations.}
  \label{tab:rocq}
\end{table}

\paragraph{Lean code volume is substantially larger, but line counts do not measure effort.}
The repository-level difference is large. The Rocq development is 34 files and 40,549 lines, whereas the root-specific part of our Lean repository that proves the same theorem is 537 files and 626,021 lines, a ratio of about fifteen to one. Within the shared BG and PF Part~I section breakdown, Table~\ref{tab:rocq} records 573,050 Lean lines and 39,163 Rocq lines. These figures describe final code volume, but they are not a normalized comparison of the work or mathematical content represented by each line. Four factors limit their interpretation.

\begin{enumerate}
\item \textbf{The base libraries differ.} A base library supplies reusable definitions and proved theorems on which a project can build. The Rocq development uses Mathematical Components, which contains extensive finite-group theory developed in large part for the Odd Order proof. The accounting boundary therefore differs between the projects: material counted as project code on one side may already be library code on the other. Subtracting the shared BG and PF Part I section totals from the corresponding development totals leaves 52,971 Lean lines and 1,386 Rocq lines outside the shared section breakdown; the Rocq remainder comprises Appendices~A and~B, the Wielandt fixed-point file, and the stripped statement. A normalized comparison would also count the relevant group-theoretic portions of both base libraries. We do not have those measurements.
\item \textbf{The proof languages differ.} Lean and Rocq use different formal languages and proof styles. In particular, \textsc{ssreflect}, the proof language used throughout Mathematical Components, is designed for concise proof scripts. A line of Lean and a line of Rocq are therefore not equivalent units of expression or work.
\item \textbf{Generated proofs may have a different internal structure.} Automatically generated Lean proofs may introduce many local helper lemmas or miss refactorings that a human-maintained project would make. The present counts cannot distinguish necessary formal infrastructure from code that later refactoring might reduce.
\item \textbf{Lines of code measure volume, not quality or effort.} Line counts record the amount and verbosity of the final code. They do not directly measure correctness, mathematical difficulty, development time, or project quality.
\end{enumerate}

Within these limits, the data support three descriptive conclusions. First, the Lean code volume for the Odd Order theorem is substantially larger. Second, the difference is distributed across the source: Lean has more lines in 16 of the 17 BG rows and all 14 PF Part~I rows, so the aggregate is not produced by one exceptional section. Third, the section-level ratios vary from $0.2$ to $34.9$, with book-level ratios of $10.0$ for BG and $19.6$ for PF Part~I. There is therefore no stable line-count conversion factor between the developments, and source length alone does not determine formalization volume, consistent with Section~\ref{sec:analysis-c1}. These measurements do not show that Lean or automatic formalization required fifteen times as much labor, that the Lean proof was mathematically harder, or that the generated proof is of lower quality than the Rocq proof.

\paragraph{Both developments check the root theorem and its logical assumptions.}
A root-theorem check asks two basic questions: (i)~\emph{whether the formal statement expresses the theorem that the project claims to prove}; and (ii)~\emph{whether its proof uses only the intended logical foundations}. Both developments address these questions, although they do so differently.

The Rocq development includes \texttt{stripped\_odd\_order\_theorem.v}, which gives a minimal, self-contained formulation of the Odd Order theorem using bare Rocq logic, without notation, coercions, or implicit arguments. The file does not merely display this formulation: it proves the stripped theorem by connecting it to \texttt{PFsection14.Feit\_Thompson}, with the resulting proof checked by the Rocq kernel. The Rocq formalization also avoids additional postulated axioms.

Our development uses \texttt{comparator} for a corresponding root-level check. It takes a statement written independently of the development, using only Mathlib constants, and verifies that the theorem proved here has the same type. It then replays the proof through the kernel and checks that its dependency closure contains no axioms beyond the three permitted by Lean's logic, namely \texttt{propext}, \texttt{Classical.choice} and \texttt{Quot.sound}. Thus, both developments guard against two important failure modes: unintentionally proving a different theorem and proving the intended theorem using unexpected assumptions.

These root-level checks do not assess the source-level justification of intermediate declarations. Translation review and reconciliation provide a separate safeguard against mistranslated statements and unsupported cross-source connections, without requiring the Lean development to reproduce the books' organization or proof routes.

\section{Ablation Studies}
\label{sec:ablation}

We conduct ablation studies of three workflow components: dependency-aware parallelism, section-level context sharing, and review-guided translation. For each component, we compare alternative configurations on a fixed target and report their effects on wall time, token consumption, or translation accuracy, as applicable. Unless stated otherwise, the experiments run on a Dell PowerEdge R750xa with Ubuntu~22.04.1, use \texttt{gpt-5.6-sol} with the reasoning-effort setting \texttt{max}, and execute through the Pi agent harness\footnote{\url{https://github.com/earendil-works/pi}} with the four tools described in Section~\ref{sec:shared-agent-runtime}. Network access is disabled, and the agents can access only the supplied Lean repository. Reported wall time includes Lean compilation: individual files can require several minutes to build, a full rebuild takes more than 30 minutes, and the agents invoke the compiler repeatedly while constructing proofs.

\subsection{Parallel versus Sequential Proof Construction}
\label{sec:ablation-parallel}
We compare the dependency-aware parallel execution of Section~\ref{sec:batch-parallelization} with a sequential schedule on the proof of Glauberman's $Z^*$ theorem. Since network access is blocked during execution, both conditions are supplied in advance with correct translations of all definitions and theorem statements; the comparison therefore isolates proof construction rather than source retrieval or statement translation. The parallel condition permits at most ten concurrent {\sc Prover} instances; individual instances run for between 0.13 and 7.6 hours. Section-level context sharing is enabled in this experiment.

\begin{table}[H]
  \centering
  \small
  \begin{tabular}{lrrr}
    \toprule
    Execution Mode & Wall Time & Total Tokens & Output Tokens \\
    \midrule
    Sequential                         & 51.4 h & 2,085,361,637 & 7,645,079 \\ 
    Parallel                           & 12.3 h & 3,038,587,634 & 9,077,309 \\ 
    Relative change (parallel vs. sequential) & $-76.1\%$ & $+45.7\%$ & $+18.7\%$ \\
    \bottomrule
  \end{tabular}
  \caption{Parallel and sequential proof construction for Glauberman's $Z^*$ theorem. Wall time includes Lean compilation.}
  \label{tab:parallel-sequential-ablation}
\end{table}

Table~\ref{tab:parallel-sequential-ablation} exhibits a direct time--compute tradeoff. Parallel execution reduces wall time from 51.4 to 12.3 hours, a 76.1\% reduction or a $4.2\times$ speedup, while increasing total token consumption by 45.7\%. Output-token consumption increases more modestly, from 7,645,079 to 9,077,309, or 18.7\%. Aggregate total-token consumption per wall-clock hour consequently rises from approximately 40.6 million to 247.4 million, reflecting the simultaneous activity of multiple {\sc Prover} instances. Parallelization is therefore effective at converting additional inference capacity into shorter elapsed time, but it is not a token-saving mechanism. The relevant benefit is critical-path compression: independent proof obligations advance concurrently, so overall completion is governed more by the slowest active branches than by the sum of all branch durations.

\subsection{Effect of Section-Level Context Sharing}
\label{sec:ablation-context-sharing}

We evaluate the section-level reuse mechanism of Section~\ref{sec:context-sharing} on PF Part~II, Chapter~4, Section~3. In the sharing condition, a single {\sc Prover} trajectory processes the theorems in sequence: after a theorem is proved, the agent session is not terminated, and the instruction for the next theorem is appended to the accumulated model context. Thus, the interaction history records source-specific notation and assumptions, earlier proof attempts, relevant local declarations, and discovered Mathlib interfaces. This model context is distinct from the Lean proof context and persists across theorem boundaries. In the no-sharing condition, each theorem instead begins in a fresh session, for 21 sessions in total. 
In both conditions, the theorems are proved sequentially without dependency-aware batch parallelization.

\begin{table}[H]
  \centering
  \small
  \begin{tabular}{lrrr}
    \toprule
    Condition & Wall Time & Total Tokens & Output Tokens \\
    \midrule
    No sharing            & 23.6 h & 520,807,040 & 923,108 \\
    Section-level sharing & 17.3 h & 426,810,415 & 730,462 \\
    Relative reduction    & 26.7\% & 18.0\% & 20.9\% \\
    \bottomrule
  \end{tabular}
  \caption{Context-sharing ablation on PF Part~II, Chapter~4, Section~3. Wall time includes Lean compilation.}
  \label{tab:context-sharing-ablation}
\end{table}

As shown in Table~\ref{tab:context-sharing-ablation}, section-level sharing reduces total token consumption by 18.0\%, output tokens by 20.9\%, and wall time by 26.5\%. In the no-sharing condition, the mathematical and formal environment of the section must be reconstructed across 21 independent sessions. The shared trajectory amortizes this discovery cost: once notation, standing assumptions, local declarations, and usable library interfaces have been established, they remain available for subsequent theorems. The consistent reductions across all three measures show that section-level context sharing is effective for this target.

\subsection{Diagnostic Feedback in Translation Review}
\label{sec:ablation-translator-review}

{\sc TranslatorAgent} uses an LLM-based {\sc ReviewAgent} as a semantic judge because Lean elaboration verifies that a proposed declaration is well typed, but not that it faithfully expresses the source statement. The judge independently compares the source material with the candidate Lean declaration; when it rejects a candidate, its feedback is returned to {\sc TranslatorAgent} for another attempt.
We conduct a pilot ablation to distinguish the effect of this review step from the effect of the information contained in its feedback. Source items were randomly sampled from the literature corpus, translated under one of four conditions, and then judged manually for fidelity to the source:
\begin{enumerate}
  \item \textbf{No review}: {\sc TranslatorAgent} returns its initial translation without an LLM review step.
  \item \textbf{Binary decision}: {\sc ReviewAgent} returns only $\mathsf{Accept}$ or $\mathsf{Reject}$.
  \item \textbf{Decision with diagnostics}: {\sc ReviewAgent} returns the binary decision together with concrete findings for the next translation attempt.
  \item \textbf{Self-improved review rubrics}: {\sc ReviewAgent}'s checklist are updated by the self-improvement process.
\end{enumerate}

\begin{table}[H]
  \centering
  \small
  \begin{tabular}{lrr}
    \toprule
    Review Feedback & Manually Judged Correct   & Tokens \\
    \midrule
    No review                    & 6/10         & 22,545,626 \\
    Binary decision              & 7/10         & 36,060,479 \\
    Decision with diagnostics    & 9/10         & 64,415,452 \\
    Self-improved                & 10/10        & 59,764,652 \\
    \bottomrule
  \end{tabular}
  \caption{Pilot ablation of review feedback for {\sc TranslatorAgent}.}
  \label{tab:translator-review-ablation}
\end{table}

The ten problems are sampled randomly from the sources and fixed for the four runs.
Table~\ref{tab:translator-review-ablation} reports six correct translations among ten items without review, seven with a binary decision, nine when diagnostics accompany the decision, and ten under the self-improved review rubric. A binary rejection prevents an unsupported candidate from advancing, but provides no direction for revision. Diagnostics additionally generate an explanation of the rejection, e.g.\ which part of the translation drifts from the source. Relative to the fixed diagnostic condition, the self-improved rubric raises the observed count from nine to ten correct translations while reducing token consumption from 64,415,452 to 59,764,652, a 7.2\% reduction. Rather than merely adding another review pass, the self-improvement process converts failures observed in earlier attempts into reusable review criteria. In this sample, the simultaneous increase in observed translation accuracy and reduction in token consumption suggests that the resulting rubric made subsequent reviews more targeted. Together with the repeated review outcomes in Section~\ref{sec:analysis-review-reconcile}, these results show that translation review serves two distinct functions---candidate selection and source-grounded revision guidance---and that the quality of the rubric affects both. The learned rubrics and their model-generated examples are shown below.

\begin{PromptBlock}{The learned checklist}

\begin{enumerate}
  \item \textbf{Recover inherited context.} When a target occurs inside a long-running
   section, standing setup, or case analysis, verify all still-active
   hypotheses and previously fixed structural data.  Do not assume that only
   conditions repeated beside the target remain in force.

  \item \textbf{Preserve canonical data.} If the source refers to maps, actions,
   subgroups, quotient actions, representatives, decompositions, or other
   objects defined earlier, ensure that the Lean declaration uses those
   canonical objects or explicitly ties parameters to them.  Do not silently
   replace canonical data with arbitrary functions, actions, or instances.

  \emph{Example.} \textcolor{gray}{For quotient group \texttt{G / H}, ensure the \texttt{MulDistribMulAction} is induced by the action on \texttt{G}. Reuse the induced quotient-action APIs rather than inventing arbitrary instances.}

  \item \textbf{Preserve complete structured data.} When the source treats a
   decomposition, model, configuration, or reconstruction datum as one
   mathematical object, verify that every distinguished component relevant to
   its meaning is represented and preserved.


  \item \textbf{Do not strengthen for implementation convenience.} Reject unnecessary
   assumptions, typeclasses, stronger action formulas, extra finiteness
   assumptions, or proof-derived properties that are not part of the source
   statement.

  \emph{Example.} \textcolor{gray}{Conclusions of theorems should not appear in assumptions.}

  \item \textbf{Audit bundled interfaces for hidden strengthening.} When reusing an
   existing project predicate, structure, or bundled interface, check that its
   hypotheses and mathematical generality match the source.

\end{enumerate}

\end{PromptBlock}

\section{Conclusion}
\label{sec:conclusion}

Large-scale source-based formalization is not only a problem of proving statements that have already been expressed in a formal language. It requires the construction of the surrounding theory from distributed sources whose dependencies, conventions, and occasional defects become visible only as formalization proceeds. We introduced {\sc FormaTheoria} to support this process end to end. The workflow combines source retrieval, semantic translation, graph-based proof construction, on-demand dependency discovery, independent review, constrained reconciliation, and human escalation, while maintaining provenance and protecting approved declarations from unsupported downstream changes.

The completed formalizations of the Feit--Thompson, Glauberman $Z^*$, Brauer--Suzuki, and Bender--Suzuki theorems demonstrate the scale at which these mechanisms can operate. The analysis reveals two distinct forms of long-range structure. Source discovery leaves the boundary of the task open-ended, whereas the declaration graph contains a broad shallow base connected to the root theorems through a narrow but exceptionally deep backbone. The execution traces provide an operational counterpart: the largest trajectories span repeated context compactions and therefore require proof state and earlier decisions to survive well beyond a single model context. Together, these observations explain why persistent graph-based state is a central workflow mechanism rather than merely a representation of the final proof.

The results also show that quality control and theory integration cannot be reduced to kernel checking alone. Review rejected the first candidate for most PF Part~I sections and turned translation into an iterative generate--check--revise process. Reconciliation modifications were concentrated at a small number of representation-heavy interfaces rather than distributed in proportion to source-group size. The source cases in Section~\ref{sec:issues} further distinguish repairs supported by an explicit mathematical bridge from conflicts that require human judgment. These findings support a division of responsibility in which Lean certifies formal correctness, review checks source fidelity, reconciliation performs only source-justified local repairs, and unresolved semantic questions remain visible for escalation.

\paragraph{Future work.} The present study remains a case study of a single large-scale formalization in finite-group theory. Our primary direction for future work is to continue formalizing the remaining components of the CFSG. Possible next targets include the Gorenstein--Walter theorem~\cite{gorenstein1965characterization,gorenstein1965characterizationII,gorenstein1965characterizationIII}, the Alperin--Brauer--Gorenstein theorem~\cite{alperin1970finite}, followed by the Gorenstein--Harada theorem~\cite{gorenstein1974finite} and the Aschbacher--Smith work~\cite{aschbacher2004quasithin}, which comprises more than one thousand pages of technically demanding mathematics. The Gorenstein--Walter theorem is already posed as a LeanEval problem~\cite{leaneval}, so a formalization of it can be checked against a statement fixed by the benchmark rather than by us. Although the current workflow provides core capabilities for dependency discovery and reconciliation, scaling to these larger targets may require more robust automated handling of source defects and cross-source misalignments. It may also require extending the batch parallelism of Section~\ref{sec:batch-parallelization} to dynamic scheduling, accelerating the project by dynamically assigning independent branches, promptly propagating approved prerequisites, prioritizing critical-path tasks, and detecting interface conflicts during branch integration. 

Two broader directions also merit investigation. One is the development of models and workflows for improving code quality and reusability by refactoring Lean developments into stable, well-documented abstractions, reducing duplicated and ad hoc constructions, and aligning reusable components with Mathlib. Because AI-generated code can be difficult for human readers to follow, such efforts would need to emphasize clear naming, transparent proof structures, and explanatory documentation. Another direction is the construction of a formalization-backed knowledge base and interactive interfaces exposing the definitions, source provenance, dependencies, and proof structure of the CFSG. The present formal library and its dependency graph could provide the underlying verified content, but realizing this possibility would require semantic annotations, links between formal declarations and their mathematical explanations, and effective tools for search, visualization, and proof navigation. Such resources could help junior mathematicians navigate the extensive prerequisites of the CFSG, enable senior mathematicians to examine its global structure more effectively, and potentially reveal alternative proof paths, reusable intermediate results, or other opportunities for new mathematical discoveries.

\section*{Acknowledgments}

We are sincerely grateful to Richard Lyons and Ronald Solomon for their valuable discussions.

\bibliographystyle{plain-titlecase}
\bibliography{refs}

\clearpage
\appendix

\section{Formalized Theorems}
\label{app:theorems}

In this section we state the theorems formalized by the workflow.
The theorems are printed here in the form in which a reader can verify them directly against
the Lean code, which is a constraint on how a statement is written and not only on how it is
proved.  Three of the four need little defending on that score.
Theorems~\ref{thm:odd-order}, \ref{thm:zstar} and~\ref{thm:brauer-suzuki} are targets of
LeanEval~\cite{leaneval}: the statement of each is posed by the benchmark, fixed in a
file a participant is forbidden to modify, and settled before any proof is attempted. Theorem~\ref{thm:brauer-suzuki} is displayed below in a simplified but equivalent form, as explained there.  It was therefore not
written by the same hands that discharged it, and could not drift towards whatever turned
out to be provable.  Two of the three are moreover stated over Mathlib and nothing else:
Theorem~\ref{thm:odd-order} needs a finite group, the parity of its order, and
solvability, and this is also, verbatim, the endpoint of our own odd order library.  Only
Theorem~\ref{thm:brauer-suzuki} adds a notion of its own, the odd core --- the supremum
of the normal subgroups of odd order --- and that is one line long.

Theorem~\ref{thm:bender-suzuki} has neither advantage.  Its statement is
ours, and one side of it must name the groups being classified, which no amount of care
will reduce to Mathlib's current vocabulary.  Definition~\ref{def:strongly-embedded} and the
discussion after the theorem are what we do about that.

\begin{theorem}[Feit--Thompson]
\label{thm:odd-order}
Let $G$ be a finite group.
If $|G|$ is odd, then $G$ is solvable.
\end{theorem}

\begin{theorem}[Glauberman's $Z^*$ theorem]
\label{thm:zstar}
Let $G$ be a finite group with an isolated involution $t$ (i.e.\ the only conjugate of $t$ commuting with $t$ is itself).
Then there is a normal subgroup $N \trianglelefteq G$ of odd order such that $[g, t] \in N$ for all $g\in G$.
\end{theorem}

\begin{theorem}[Brauer--Suzuki]
\label{thm:brauer-suzuki}
Let $G$ be a finite group whose Sylow $2$-subgroups are generalized
quaternion groups of order $2^{n}$ with $n \geq 3$, and let $t$ be an involution of $G$.  Write
$O(G)$ for the odd core of $G$, the largest normal subgroup of odd order.  Then the image
of $t$ in $G/O(G)$ is central.
\end{theorem}

Theorem~\ref{thm:brauer-suzuki} is not the benchmark statement verbatim.  What
LeanEval poses, at commit \texttt{ff6a35e}, is the following.

\begingroup\color{blue}
\noindent\begin{minipage}{\linewidth}
\begin{lstlisting}
theorem brauer_suzuki {G : Type*} [Group G] [Finite G]
    (n : ℕ) (hn : 3 ≤ n)
    (P : Sylow 2 G)
    (hquat : Nonempty ((P : Subgroup G) ≃* QuaternionGroup (2 ^ (n - 2))))
    (t : G) (ht_mem : t ∈ (P : Subgroup G)) (ht_ord : orderOf t = 2) :
    (QuotientGroup.mk t : G ⧸ oddCore G) ∈ Subgroup.center (G ⧸ oddCore G)
\end{lstlisting}
\end{minipage}
\endgroup

Theorem~\ref{thm:brauer-suzuki} is this statement with \texttt{ht\_mem}, the
requirement that $t$ lies in the named Sylow $2$-subgroup $P$, left out, i.e.\ our proof of this statement does not use the \texttt{ht\_mem} hypothesis.  The two are
equivalent, and in the direction that matters ours is the stronger: an involution lies in
some Sylow $2$-subgroup, all of them are conjugate, and the conclusion never mentions the
one chosen, so the hypothesis is recovered by replacing $P$ with a conjugate.
The statement is equivalent to
the more familiar one that $Z(G/O(G))$ has order exactly two.
The equivalence is straightforward.
The center of $G/O(G)$ has no non-trivial subgroup of odd order, since such a subgroup would be a normal subgroup of odd order in a group
with none.
Thus, the center is a $2$-group and it lies inside every Sylow $2$-subgroup; being central
it lies in that subgroup's own center, which has order two because the subgroup is
generalized quaternion.  The image of $t$ is then its non-trivial element.
Conversely, a center of order two visibly contains a central involution.

The last of the four theorems is stated in terms of the notion of a strongly embedded subgroup.
We set it out on its own because it carries the weight of everything said
after it.  The notion is built from Mathlib primitives alone --- subgroups, conjugation,
and whether an element is an involution --- so a reader can satisfy themselves that the
Lean code matches what the literature says without leaving Mathlib's vocabulary.  That is what
makes it useful here: it is the connective between a statement that is easy to vet as
correctly formalized and the harder one, in terms of matrix groups, that the
classification actually delivers.

\begin{definition}[Strongly embedded subgroup]
\label{def:strongly-embedded}
A subgroup $M$ of a finite group $X$ is \emph{strongly embedded} if $M$ is
proper, $M$ contains an involution, and $M \cap M^{g}$ contains no involution for every
$g \in X \setminus M$.
\end{definition}

In Lean this is the whole of it:

\noindent\begin{minipage}{\linewidth}
\begin{lstlisting}
def IsInvolution {G : Type*} [Group G] (x : G) : Prop := x ≠ 1 ∧ x ^ 2 = 1

def IsStronglyEmbedded {X : Type*} [Group X] (M : Subgroup X) : Prop :=
  M ≠ ⊤ ∧ (∃ x ∈ M, IsInvolution x) ∧
    ∀ g ∉ M, ∀ x ∈ M ⊓ M.map (MulAut.conj g).toMonoidHom, ¬ IsInvolution x
\end{lstlisting}
\end{minipage}

\begin{theorem}[Bender--Suzuki]
\label{thm:bender-suzuki}
For a finite simple group $X$ the following are equivalent:
\begin{enumerate}
    \item\label{item:strongly-embedded} $X$ has a strongly embedded subgroup;
    \item\label{item:three-families} $X$ is isomorphic to $\mathrm{PSL}_2(2^{n})$ with $n \geq 2$, to
      the Suzuki group $\mathrm{Sz}(2^{2n+1})$ with $n \geq 1$, or to
      $\mathrm{PSU}_3(2^{n})$ with $n \geq 2$.
\end{enumerate}
\end{theorem}

Two features of this statement are included primarily for the benefit of the reader. They illustrate what a formalization of a theorem of this scale can provide beyond the bare fact that Lean accepts its proof.

The first feature is that the statement is an equivalence, rather than only the classification implication from Side~\ref{item:strongly-embedded} to Side~\ref{item:three-families} of Theorem~\ref{thm:bender-suzuki}. The two sides differ considerably in how readily a human reader can inspect them. Side~\ref{item:three-families} names three families of matrix groups, two of them constructed here. Fully checking this side requires examining the underlying constructions, particularly that of \(\mathrm{Sz}(2^{2n+1})\), which is elaborate and has no existing counterpart in Mathlib. Side~\ref{item:strongly-embedded}, by contrast, is Definition~\ref{def:strongly-embedded}. It fits on three lines and can be compared directly with the standard definition in the group-theory literature.

The reverse implication from Side~\ref{item:three-families} to Side~\ref{item:strongly-embedded} therefore serves as a reader-facing sanity check on the formalization. It verifies that every group belonging to one of the three explicitly constructed families has a strongly embedded subgroup: respectively, a Borel subgroup, the stabilizer of a point of the ovoid, or the stabilizer of an isotropic point. This direction is not needed to state the classification, and proving it directly is itself highly nontrivial. Its formal proof consequently establishes a substantial and readily intelligible mathematical consequence of the constructions, even for readers who do not inspect every detail of their matrix representations. Together with the classification direction, it also shows that the three constructed families collectively coincide with the finite simple groups possessing a strongly embedded subgroup.

This sanity check supplements the identification of the constructions. Our \(\mathrm{PSL}_2(2^n)\) is\\
\texttt{ProjectiveSpecialLinearGroup (Fin 2) (GaloisField 2 n)}, namely Mathlib's projective special linear group over its field of order \(2^n\). Our \(\mathrm{PSU}_3(2^n)\) is defined, using Mathlib's matrix vocabulary, as the group of determinant-one isometries of the Hermitian form with antidiagonal Gram matrix, modulo scalars. Our construction of \(\mathrm{Sz}(2^{2n+1})\) is based directly on Suzuki's original construction in~\cite{suzuki1960new}. Following that construction, we define it as the subgroup of \(\mathrm{GL}_4(2^{2n+1})\) generated by three explicit families of matrices. The correspondence is therefore grounded in the published matrix construction itself. Because this construction is substantially more involved than those of the other two families, however, checking every detail requires considerable effort from a human reader. The equivalence provides an additional check by formally establishing that these explicitly constructed groups possess the highly nontrivial structural property required by the Bender--Suzuki theorem.

The second feature is that Side~\ref{item:strongly-embedded} contains essentially no project-specific group-theoretic machinery. It is stated using Mathlib's notions of subgroups, conjugation, and lattice operations, together with a one-line definition of an involution given four lines earlier in the same file. In the development, this side is represented by \texttt{IsStronglyEmbedded} in
\texttt{BenderSuzuki/FinalTheorem.lean}. This is the same predicate that appears as the hypothesis of \texttt{bender\_suzuki}, and the equivalence is formalized as
\texttt{isSimpleBenderGroup\_\allowbreak iff\_\allowbreak exists\_\allowbreak stronglyEmbedded}, in the file \texttt{Classification.lean}. Thus, the readily inspectable side of the equivalence is stated almost entirely in standard Mathlib vocabulary and can be compared directly with the conventional mathematical formulation.

Reading a statement is one check; the other is mechanical.  Lean accepts any file that compiles, and macros, custom tactics and new axioms give an author many tricks to compile something that does not prove what its name suggests.  The comparator~\cite{comparator} closes that gap: it checks that the theorem proved here has the same type as the statement written independently of this development, then replays the proof through the kernel, admitting no axioms beyond the three permitted ones.  Section~\ref{sec:analysis-rocq} describes how we use it.

\section{Prompts}

\begin{PromptBlock}{{\sc ProveAgent}}

\begin{PromptBlock}{Using Lean compiler}
The following instructions focus on how to use the Lean compiler to test Lean interfaces and write test snippets.
\end{PromptBlock}

Use a disposable Lean file as the scratch surface for an investigation. Prefer a uniquely named directory under \texttt{/tmp}, created with \texttt{mktemp -d}, so probing does not pollute the repository. If the user designates an existing probe file, inspect it before editing and preserve unrelated work.

\paragraph{Probe loop.}
\begin{enumerate}
\item Import the narrowest module that exposes the declarations under study.
\item Keep related checks and small examples together while the investigation is active.
\item Run the probe from the repository root:
\begin{lstlisting}
    lake env lean /tmp/<probe-directory>/Probe.lean
\end{lstlisting}
\item Read the complete diagnostic, revise one uncertainty, and rerun.
\item Put the final proof or definition in its owning module. Remove only temporary artifacts that you created.
\end{enumerate}

A successful probe establishes elaboration in the probe environment; it does not replace building the edited source module.

\paragraph{Inspect declarations.}

Use exact qualified names when possible:

\noindent\begin{minipage}{\linewidth}
\begin{lstlisting}
import BenderSuzuki.PFchapter4section1.lemma_a

#check Subgroup.map_normalClosure
#check @Subgroup.map_normalClosure
#print Subgroup.map_normalClosure
#print axioms BenderSuzuki.PFchapter4section1.lemma_a

set_option pp.all true in
#check Subgroup.map_normalClosure
\end{lstlisting}
\end{minipage}

\begin{itemize}
    \item Use \texttt{\#check} for the elaborated type.
    \item Add \texttt{@} to expose implicit arguments.
    \item Use \texttt{\#print} for the declaration body or theorem statement available through the environment.
    \item Use \texttt{\#print axioms} only after rebuilding edited source. Treat \texttt{sorryAx} as unresolved proof debt.
    \item Scope \texttt{set\_option pp.all true} with \texttt{in} to expose coercions, universes, reducible wrappers, and synthesized arguments without flooding the entire probe.
\end{itemize}

\paragraph{Inspect proof states.}

Write a small \texttt{example} with the same imports, variables, hypotheses, and target:

\noindent\begin{minipage}{\linewidth}
\begin{lstlisting}
example {G : Type*} [Group G] (H : Subgroup G) : H ≤ H := by
  intro x hx
  trace_state
  exact hx
\end{lstlisting}
\end{minipage}

Use \texttt{trace\_state} at the point of interest. Deliberately leaving a disposable example incomplete is also useful because the CLI prints the remaining context and goal. A temporary \texttt{sorry} is permitted only in the disposable probe; never count it as progress or copy it into finished source.

Test candidate terms in separate labeled examples or one at a time. Preserve the exact type mismatch from a failed candidate instead of retrying it without a new reason.

\paragraph{Diagnose source and freshness.}

\begin{itemize}
    \item Run \texttt{lake env lean path/to/Target.lean} for direct source diagnostics.
    \item Convert a source path to its build target by dropping \texttt{.lean} and replacing \texttt{/} with \texttt{.}, then run \texttt{lake build <Target.Module>}.
    \item Remember that importing `Target` loads its last built \texttt{.olean}; it does not expose unbuilt edits to \texttt{Target.lean}.
    \item After changing source, rebuild the owning module before checking an exported declaration or its axioms through a probe.
    \item Capture long output under \texttt{/tmp} and scan it with
    \begin{lstlisting}
        rg "error:|warning:|sorry|unsolved goals"\end{lstlisting}
    while retaining the full log for context.
\end{itemize}

When an error appears only in a downstream module, compare the exported declaration's visibility, import boundary, and exact elaborated type before changing the proof itself.

\begin{PromptBlock}{Finding a proof}
The following instructions focus on how to find a proof of a theorem directly.
\end{PromptBlock}

\paragraph{Establish the proof contract.}
\begin{enumerate}
    \item Restate the target in plain mathematical language: inputs, typeclass assumptions, hypotheses, and conclusion.
    \item Inspect the surrounding namespace, variables, imports, visibility, and neighboring proof style.
    \item Do not silently weaken, strengthen, or replace the claim.
    \item Build the owning module before editing when a clean baseline is important.
\end{enumerate}

\paragraph{Loop for a proof.}
\begin{enumerate}
    \item Normalize only enough to expose the mathematical shape with \texttt{dsimp}, \texttt{simp}, or carefully chosen rewrites.
    \item Search the local repository, then Mathlib, before proving reusable infrastructure from scratch.
    \item  Confirm promising declarations with \texttt{\#check} or a small probe using the target imports.
    \item Decompose the proof into the smallest useful helpers. Keep theorem-local facts local or \texttt{private}; promote a helper to public API only when a real downstream consumer justifies it.
    \item Make a small source edit and run \texttt{lake build <Target.Module>}.
    \item Read the complete first error, inspect the local goal in a probe when needed, and fix one cause at a time.
    \item Once the owner builds, run the narrowest relevant downstream build. Use a full \texttt{lake build} for final integration when the change crosses libraries or build configuration.
\end{enumerate}

Prefer disposable \texttt{example} declarations over inserting exploratory placeholders into production source. If a top-down skeleton in source is genuinely useful, keep every \texttt{sorry} short-lived and explicitly tracked, and remove all new placeholders before completion.

\paragraph{Proof construction priorities.}
Use tactics in roughly this order, choosing the clearest fit rather than forcing a fixed script:
\begin{itemize}
    \item Simplification: \texttt{simp}, \texttt{simp?}, \texttt{dsimp}, \texttt{simp only}, or \texttt{simpa}.
    \item Structure: \texttt{intro}, \texttt{refine}, \texttt{constructor}, \texttt{use}, \texttt{rcases}, \texttt{obtain}, and \texttt{cases}.
    \item Rewriting and transport: \texttt{rw}, \texttt{nth\_rw}, \texttt{simp\_rw}, \texttt{change}, \texttt{show}, \texttt{convert}, and \texttt{calc}.
    \item Domain automation: \texttt{ring}, \texttt{nlinarith}, \texttt{linarith}, \texttt{omega}, or other imported tactics when the goal matches their domain.
    \item Search automation: \texttt{exact?}, \texttt{apply?}, \texttt{aesop?}, and \texttt{grind} for exploration; retain them when stable and readable, otherwise use their suggestions to write explicit steps.
\end{itemize}

Prefer a short \texttt{calc} or named intermediate fact when it exposes the mathematical reason for a transformation. Avoid large undirected automation when a few deterministic steps make dependencies and failure modes clearer.

\paragraph{Decompose hard proofs.}

Before a new difficult route, write a 3--8 line sketch naming:

\begin{itemize}
    \item the reduction from the main goal;
    \item the helper declarations or local facts;
    \item the key project or Mathlib dependencies;
    \item the highest-risk inference.
\end{itemize}

Discharge easy normality, containment, equality, and impossible branches early. Isolate the hard core as an explicit lemma with the weakest useful statement. Make the main theorem typecheck against helper statements when that clarifies the architecture, then prove helpers from easiest to hardest.

If an informal source makes an unsupported jump, isolate the exact implication and try to falsify it under weaker assumptions. Repair the hypothesis or route, prove a narrow transfer lemma, or record a blocker; do not encode the missing inference as an axiom, opaque declaration, or unexplained \texttt{sorry}.

\paragraph{Debug common failures.}

\begin{itemize}
    \item If simplification stalls, use \texttt{simp?}, inspect the suggested lemmas, then prefer a focused \texttt{simp only [...]} when stability matters.
    \item If unification fails, inspect \texttt{\#check @name}, make key arguments explicit, and check coercions or orientation.
    \item If typeclass search fails or loops, identify the missing instance. Add it as a hypothesis or local \texttt{haveI} when mathematically justified; avoid an unnecessary global instance.
    \item If the goal looks unrelated, unfold only the relevant definitions with \texttt{dsimp}, \texttt{change}, or \texttt{simp [definition]} and confirm that the statement has not drifted.
    \item If a definition in a \texttt{module} file cannot unfold, add \texttt{@[expose]} only when unfolding is part of its intended interface; otherwise prove and use an exposed helper theorem.
    \item If a downstream file cannot see a declaration, check \texttt{public} visibility and whether the dependency requires \texttt{public import}.
    \item If a proof works only with a broader import, locate the actual dependency and decide whether the target should import it or the helper belongs in a lower module.
\end{itemize}

\paragraph{Verify completion.}

\begin{enumerate}
    \item Build the owning module and relevant downstream target.
    \item Inspect logs for \texttt{error:}, \texttt{warning:}, \texttt{sorry}, and \texttt{unsolved goals}; explain any pre-existing warning that remains relevant.
    \item Scan changed Lean source for unintended \texttt{sorry}, \texttt{axiom}, and \texttt{opaque} declarations.
    \item Rebuild before using \texttt{\#print axioms Fully.Qualified.name} on exported theorems.
    \item Confirm intended names, namespace, visibility, imports, and statement meaning.
    \item Run \texttt{git diff --check} and preserve unrelated working-tree changes.
\end{enumerate}

Do not declare the task complete merely because a scratch example elaborates or one local theorem closes; satisfy the user's full requested scope and integration boundary.

\begin{PromptBlock}{Long-horizon management}
The following instructions focus on how to manage the proof task using a DAG.
\end{PromptBlock}

Since proving the targe theorem is likely to span multiple context window and sessions, maintain files for tracking states.
Update live task state incrementally during the task.

\paragraph{Proof state management.}
In \texttt{state.md}, record the following information:
\begin{itemize}
    \item \emph{active target}, the theorem being proved;
    \item \emph{progress}, the actions taken by the agent;
    \item \emph{lessons}, reusable experience for next steps;
    \item \emph{validation}, the acceptance gate of the target.
\end{itemize}
Keep the state file up to date.

\paragraph{Proof graph management.}
A node is a theorem, or definition whose statement is formalized or introduced as a helper.
A route is a tentative proof decomposition that may be refined, or pruned.
Node states are only \texttt{proposed}, \texttt{proved}, and \texttt{pruned}.
If no usable route exists, inspect the Lean statement, source theorem/proof, and dependencies, then write a short Lean-oriented proof sketch and helper plan.
If the route is blocked, choose one: split the current subnode, replace a helper, or prune the route and create a new one.

\end{PromptBlock}

\begin{PromptBlock}{{\sc RefineAgent}}
The task is to refine the node graph and the active route.
A node is a theorem, or definition whose statement is formalized or introduced as a helper.
A route is a tentative proof decomposition that may be refined, or pruned.
Node states are only \texttt{proposed}, \texttt{proved}, and \texttt{pruned}.

The current active target \{\{lean target\}\} is too hard to be proved directly.
Please decompose it into several helper lemmas and build a new route.
\begin{itemize}
    \item Isolate the exact inference, list the available hypotheses, and try to falsify the bare implication in a smaller setting.
    \item Before a new route, write a proof sketch: goal, helpers, dependencies, and risk.
    \item Prune easy branches early. Discharge normal, containment, and impossible branches first, then isolate the hard branch.
\end{itemize}

\end{PromptBlock}

\begin{PromptBlock}{{\sc TranslatorAgent}}
Your task is to translate all definitions and theorem statements in \{\{\texttt{source item}\}\} in \{\{\texttt{Lean file}\}\}.

\paragraph{Requirements.}
\begin{itemize}
    \item You only need to formalize the statements of definitions and theorems etc., and fill the proof bodies with sorry.
    \item Every definition and theorem must be exactly the same as the natural language version in the source.
    In particular, you should make sure the Lean version doesn't omit assumptions in the source.
\end{itemize}
\end{PromptBlock}

\begin{PromptBlock}{{\sc ReviewAgent} for {\sc TranslatorAgent}}
The task is to review \{\{lean file\}\} for whether it is a source faithful translation of \{\{source item\}\}.

The original task instruction is

\noindent\begin{minipage}{\linewidth}
\begin{lstlisting}
```
{{TranslatorAgent instructions}}
```
\end{lstlisting}

\paragraph{Common pitfalls.}
The following is a list of common mistakes one usually makes when translating definitions and statements to Lean.
Pay attention to the list:
\{\{learned checklist\}\}

\end{minipage}
\end{PromptBlock}

\begin{PromptBlock}{{\sc ReconcilerAgent}}
Currently we need to invoke \{\{lean target\}\} in \{\{lean source\}\} as describe by the natural-language \{\{source items\}\}.
Please investigate the situation and
\begin{itemize}
    \item fix without changing downstream definitions or theorem statements;
    \item request human investigation.
\end{itemize}

\paragraph{Rules for fixing.}
When attempting to fix the given lean code, follow the rules:
\begin{itemize}
    \item You may change the given lean code, definitions, theorem statements, or proofs.
    \item The final Lean repository must build.
    In particular, if you changed a given lean declaration, all downstream usage of the declaration must be updated.
    \item You're not allowed to modify any definitions or statements except the two given to you.
    \item If you fails to fix the lean code, report to human immediately.
\end{itemize}

\paragraph{Common patterns.}
\begin{itemize}
    \item If the theorem statement asserts existence, try to make the statement constructive, i.e.\ write down the concrete construction in the statement.
    Consider this rule in particular when downstream needs uniqueness of the existence.
    \item The universe level parameters should be as general as possible.
    \item If the desired Lean theorem cannot be invoked due to mismatched type, first write down a separate bridge lemma deriving the desired property from the existing assumptions.
\end{itemize}

\end{PromptBlock}

\end{document}